%% file: 3pol.tex
\documentclass[a4paper,USenglish]{lipics-v2016}
\usepackage{3pol}

\title{Subquadratic Algorithms for Algebraic Generalizations of 3SUM}

\author[1]{Luis Barba}
\author[2]{Jean Cardinal\thanks{Supported by the ``Action de Recherche Concert\'ee'' (ARC) COPHYMA, convention number 4.110.H.000023.}}
\author[3]{John Iacono\thanks{Research partially completed while on sabbatical at the
Algorithms Research Group of the
D\'{e}partement d'Informatique at the
Universit\'{e} Libre de Bruxelles with support from
a Fulbright Research Fellowship,
the Fonds de la Recherche Scientifique --- FNRS,
and NSF grants CNS-1229185, CCF-1319648, and CCF-1533564.}}
\author[2]{Stefan Langerman\thanks{Directeur de recherches du F.R.S.-FNRS}}
\author[2]{Aur\'elien Ooms\thanks{Supported by the Fund for Research Training in Industry and Agriculture (FRIA).}}
\author[4]{Noam Solomon\thanks{Supported by Grant 892/13 from the
Israel Science Foundation}}
\affil[1]{Department of Computer Science, ETH Zürich, Switzerland.\\
\texttt{luis.barba@inf.ethz.ch}}
\affil[2]{D\'epartement d'Informatique, Universit\'e libre de Bruxelles (ULB), Belgium.\\
\texttt{\{jcardin,slanger,aureooms\}@ulb.ac.be}}
\affil[3]{Department of Computer Science and Engineering, New York University (NYU), USA.\\
\texttt{socg2017@johniacono.com}}
\affil[4]{School of Computer Science, Tel Aviv University (TAU), Israel.\\
\texttt{noam.solom@gmail.com}}%

\authorrunning{L.\,Barba, J.\,Cardinal, J.\,Iacono, S.\,Langerman, A.\,Ooms and N.\,Solomon}
\Copyright{Luis Barba, Jean Cardinal, John Iacono, Stefan Langerman, Aur\'elien Ooms and Noam Solomon}
\EventEditors{John Q. Open and Joan R. Acces}
\EventNoEds{2}
\EventLongTitle{42nd Conference on Very Important Topics (CVIT 2016)}
\EventShortTitle{CVIT 2016}
\EventAcronym{CVIT}
\EventYear{2016}
\EventDate{December 24--27, 2016}
\EventLocation{Little Whinging, United Kingdom}
\EventLogo{}
\SeriesVolume{42}
\ArticleNo{23}

\begin{document}
\maketitle

\input{sec/00-abstract}
\input{sec/01-introduction}

\input{sec/02-definitions-and-previous-work}
\input{sec/03-algorithms/01-nonuniform-explicit}
\input{sec/03-algorithms/03-uniform-explicit}
\input{sec/03-algorithms/04-dominance-reporting}
\input{sec/03-algorithms/99-3pol}
{\small%
\bibliography{3pol}}
\appendix
\input{sec/03-algorithms/02-batch-range-searching}
\input{03-algorithms-04-dominance-reporting-analysis}
\input{sec/03-algorithms/05-nonuniform-implicit}

\input{sec/03-algorithms/06-uniform-implicit}
\input{sec/04-applications/00-intro}
\input{sec/04-applications/01-gpt}

\input{sec/04-applications/02-unit-circles}
\input{sec/04-applications/03-unit-triangles}

\end{document}

%% file: sec/00-abstract.tex
\begin{abstract}
The 3SUM problem asks if an input $n$-set of real numbers contains a triple
whose sum is zero. We consider the 3POL problem, a natural generalization of
3SUM where we replace the sum function by a constant-degree polynomial in three
variables.
The motivations are threefold.
Raz, Sharir, and de Zeeuw gave an $O(n^{11/6})$ upper bound
on the number of solutions of
trivariate polynomial equations when the solutions are taken from the cartesian product
of three $n$-sets of real numbers. We give algorithms for the corresponding
problem of counting such solutions. Gr\o nlund and Pettie
recently designed subquadratic algorithms for 3SUM\@. We generalize their
results to 3POL\@. Finally, we shed light on the General
Position Testing~(GPT) problem: ``Given $n$ points in the plane, do three of
them lie on a line?'', a key problem in computational geometry.

We prove that there exist bounded-degree algebraic decision trees of depth
$O(n^{\frac{12}{7}+\varepsilon})$ that solve 3POL, and that 3POL can be solved
in $O(n^2 {(\log \log n)}^\frac{3}{2} / {(\log n)}^\frac{1}{2})$ time in the
real-RAM model.
Among the possible applications of those results, we show how to solve GPT in
subquadratic time when the input points lie on $o({(\log n)}^\frac{1}{6}/{(\log
\log n)}^\frac{1}{2})$ constant-degree polynomial curves.
This constitutes the first step towards closing the major open question of
whether GPT can be solved in subquadratic time.
To obtain these results, we generalize important tools --- such as batch range
searching and dominance reporting --- to a polynomial setting. We expect
these new tools to be useful in other applications.
\end{abstract}

%% file: sec/01-introduction.tex
\section{Introduction}

The 3SUM problem is defined as follows: given $n$ distinct real numbers, decide
whether any three of them sum to zero.
A popular conjecture is that no $O(n^{2-\delta})$-time algorithm for 3SUM
exists. This conjecture has been used to show conditional lower bounds for
problems in P, notably in computational geometry with problems such as
GeomBase, general position~\cite{GO95} and Polygonal Containment~\cite{BH01},
and more recently for string problems such as Local Alignment~\cite{AVW14} and
Jumbled Indexing~\cite{ACLL14}, as well as dynamic versions of graph
problems~\cite{P10,AV14}, triangle enumeration and Set
Disjointness~\cite{KPP16}.
For this reason, 3SUM is considered one of the key subjects of an
emerging theory of complexity-within-P, along with other problems such as
all-pairs shortest paths,
orthogonal vectors,
boolean matrix multiplication,
and conjectures such as the Strong Exponential Time
Hypothesis~\cite{AVY15,HKNS15,CGIMPS16}.

Because fixing two of the numbers $a$ and $b$ in a triple only allows for one
solution to the equation $a + b + x = 0$, an instance of 3SUM has at most
$n^2$ solution triples. An instance with a matching lower bound is for example
the set $\{\,\frac{1-n}{2},\ldots,\frac{n-1}{2}\,\}$ (for odd
$n$) with $\frac{3}{4} n^2 + \frac 14$ solution triples.
One might be tempted to think that the number of solutions to the problem
would lower bound the complexity of algorithms
for the decision version of the problem, as it is the case for restricted
models of computation~\cite{E99}.
This is a common misconception.
Indeed, Gr\o nlund and Pettie~\cite{GP14} recently proved that there exist
$\tilde{O}(n^{3/2})$-depth linear decision trees and $o(n^2)$-time real-RAM
algorithms for 3SUM.

A natural generalization of the 3SUM problem is to replace the sum function by
a constant-degree polynomial in three variables $F \in \mathbb{R}[x,y,z]$ and
ask to determine whether there exists any triple $(a,b,c)$ of input numbers
such that $F(a,b,c)=0$. We call this new problem the \emph{3POL problem}.

For the particular case $F(x,y,z) = f(x,y) - z$ where $f \in \mathbb{R}[x,y]$
is a constant-degree bivariate polynomial, Elekes and Rónyai~\cite{ER00} show
that the number of solutions to the 3POL problem is $o(n^2)$ unless $f$ is
\emph{special}. Special for $f$ means that $f$ has one of the two special forms
$f(u,v)=h(\varphi(u)+\psi(v))$ or $f(u,v)=h(\varphi(u)\cdot\psi(v))$, where
$h,\varphi,\psi$ are univariate polynomials of constant degree. Elekes and
Szabó~\cite{ES12} later generalized this result to a broader range of
functions $F$ using a wider
definition of specialness. Raz, Sharir and Solymosi~\cite{RSS14} and Raz,
Sharir and de Zeeuw~\cite{RSZ15} recently improved both bounds
on the number of solutions to
$O(n^{11/6})$.
They translated the problem into an incidence problem between points and
constant-degree algebraic curves. Then, they showed that unless $f$ (or $F$) is
special, these curves have low multiplicities. Finally, they applied a theorem
due to Pach and Sharir~\cite{PS98} bounding the number of incidences
between the points and the curves. Some of these ideas appear in our approach.

In computational geometry, it is customary to assume the real-RAM model can be
extended to allow the computation of roots of constant degree polynomials.
We distance ourselves from this practice and take particular care
of using the real-RAM model and the bounded-degree algebraic decision tree
model with only the four arithmetic operators.

\subsection{Our results}
We focus on the computational complexity of 3POL\@. Since 3POL contains 3SUM, an
interesting question is whether a generalization of Gr\o nlund and
Pettie's 3SUM algorithm exists for 3POL\@. If this is true, then we might wonder
whether we can beat the $O(n^{11/6}) = O(n^{1.833\ldots})$ combinatorial bound of Raz, Sharir and de
Zeeuw~\cite{RSZ15} with nonuniform algorithms. We give a positive answer to
both questions: we show there exist $O(n^2 (\log \log n)^\frac{3}{2} / \log
n)^\frac{1}{2})$-time real-RAM algorithms and
$O(n^{12/7+\varepsilon}) = O(n^{1.7143})$-depth bounded-degree algebraic
decision trees for 3POL\@.\footnote{Throughout this document, $\varepsilon$
denotes a positive real number that can be made as small as desired.}
To prove our main result, we present a fast algorithm for
the Polynomial Dominance Reporting (PDR) problem, a far reaching generalization of
the Dominance Reporting problem. As the algorithm for Dominance Reporting and
its analysis by Chan~\cite{Cha08} is used in fast algorithms for all-pairs
shortest paths, (min,+)-convolutions, and 3SUM, we expect this new algorithm
will have more applications.

Our results can be applied to many degeneracy testing problems, such
as the General Position Testing~(GPT) problem: ``Given $n$ points in the plane, do
three of them lie on a line?'' It is well known that GPT is 3SUM-hard,
and it is open whether GPT admits a subquadratic algorithm. Raz, Sharir
and de Zeeuw~\cite{RSZ15} give a combinatorial bound of $O(n^{11/6})$ on the
number of collinear triples when the input points are known to be lying on
a constant number of polynomial curves, provided those curves are neither
lines nor cubic curves. A corollary of our first result is that
GPT where the input points are constrained to lie on
$o((\log n)^\frac{1}{6}/(\log \log n)^\frac{1}{2})$
constant-degree polynomial curves (including lines and cubic curves)
admits a subquadratic real-RAM algorithm and
a strongly subquadratic bounded-degree algebraic decision tree.
Interestingly, both reductions from 3SUM to GPT on 3 lines (map $a$ to $(a,0)$,
$b$ to $(b,2)$, and $c$ to $(\frac c2, 1)$) and from 3SUM to GPT on a
cubic curve (map $a$ to $(a^3,a)$, $b$ to $(b^3,b)$, and $c$ to $(c^3,c)$)
construct such special instances of GPT\@.
This constitutes the first step towards closing the major open question of
whether GPT can be solved in subquadratic time.
This result is described in Appendix~\ref{sec:applications} where
we also explain how to apply our algorithms to the problems of counting
triples of points spanning unit circles or triangles.

%% file: sec/02-definitions-and-previous-work.tex
\subsection{Definitions}

\subparagraph{3POL}
We look at two different generalizations of 3SUM\@. In the first one, which we
call 3POL, we replace the sum function by a trivariate polynomial of constant
degree.
\begin{problem}[3POL]
Let $F \in \mathbb{R}[x,y,z]$ be a trivariate polynomial of constant degree,
given three sets $A$, $B$, and $C$, each containing $n$ real numbers, decide
whether there exist $a \in A$, $b \in B$, and $c \in C$ such that
$F(a,b,c)=0$.
\end{problem}
The second one is a special case of 3POL where we restrict the trivariate
polynomial $F$ to have the form $F(a,b,c) = f(a,b) - c$.
\begin{problem}[explicit 3POL]
Let $f \in \mathbb{R}[x,y]$ be a bivariate polynomial of constant degree,
given three sets $A$, $B$, and $C$, each containing $n$ real numbers, decide
whether there exist $a \in A$, $b \in B$, and $c \in C$ such that $c=f(a,b)$.
\end{problem}
We look at both uniform and nonuniform algorithms for explicit 3POL and 3POL\@.
We begin with an $O(n^{\frac{12}{7}+\varepsilon})$-depth bounded-degree
algebraic decision tree for explicit 3POL in
\S\ref{sec:algo:explicit:nonuniform}.
In \S\ref{sec:algo:explicit:uniform}, we continue by giving a similar real-RAM
algorithm for explicit 3POL that achieves subquadratic running time.
In Appendix~\ref{sec:algo:implicit:nonuniform}, we go back to the
bounded-degree algebraic decision tree for explicit 3POL and generalize it to
work for 3POL\@.
Finally, in Appendix~\ref{sec:algo:implicit:uniform}, we give a similar real-RAM
algorithm for 3POL that runs in subquadratic time.

\subparagraph{Models of Computation}

Similarly to Gr\o nlund and Pettie~\cite{GP14}, we consider both nonuniform
and uniform models of computation.
For the nonuniform model, Gr\o nlund and Pettie consider linear
decision trees, where one is only allowed to manipulate the input numbers
through linear queries to an oracle. Each linear query has constant cost and
all other operations are free but cannot inspect the input.
In this paper, we consider
\emph{bounded-degree algebraic decision trees (ADT)}~\cite{R72,Y81,SY82},
a natural generalization of linear decision trees,
as the nonuniform model. In a bounded-degree algebraic decision tree, one
performs constant cost branching operations that amount to test the sign of
a constant-degree polynomial for a constant number of input numbers. Again,
operations not involving the input are free.
For the uniform model we consider the real-RAM model with only the four
arithmetic operators.

The problems we consider require our algorithms to manipulate polynomial
expressions and, potentially, their real roots. For that purpose, we will rely
on Collins cylindrical algebraic decomposition (CAD)~\cite{C75}.
To understand the power of this method, and why it is useful for us, we give some
background on the related concept of first-order theory of the reals.

\begin{definition}
	A Tarski formula $\phi \in \mathbb{T}$ is a grammatically correct formula consisting of real
	variables ($x \in \mathbb{X}$), universal and existential quantifiers on those real
	variables ($\forall,\exists\colon\,\mathbb{X}\times\mathbb{T}\to\mathbb{T}$),
	the boolean operators of conjunction and
	disjunction ($\land,\lor\colon\,\mathbb{T}^2\to\mathbb{T}$), the six
	comparison operators ($<,\le,=,\ge,>,\ne\colon\,\mathbb{R}^2\to\mathbb{T}$),
	the four arithmetic operators ($+,-,*,/\colon\,\mathbb{R}^2\to\mathbb{R}$),
	the usual parentheses that modify the priority of operators, and constant
	real numbers.
	A Tarski sentence is a fully quantified Tarski formula.
	The first-order theory of the reals (\FOTR{}) is
	the set of true Tarski sentences.
\end{definition}

Tarski~\cite{T51} and Seidenberg~\cite{Sei74} proved that \FOTR{} is decidable.
However, the algorithm resulting from their proof has nonelementary complexity.
This proof, as well as other known algorithms, are based on quantifier
elimination, that is, the translation of the input formula to a much longer
quantifier-free formula, whose validity can be checked.
There exists a family of formulas for which any method of quantifier elimination
produces a doubly exponential size quantifier-free formula~\cite{DH88}.
Collins CAD matches this doubly exponential complexity.
\begin{theorem}[Collins~\cite{C75}]
	\FOTR{} can be solved in $2^{2^{O(n)}}$ time.
\end{theorem}

See
Basu, Pollack, and Roy~\cite{BPR06} for additional details,
Basu, Pollack, and Roy~\cite{BPR96b} for a singly exponential algorithm when all
quantifiers are existential
(existential theory of the reals, \ETR{}),
Caviness and Johnson~\cite{CJ12} for an anthology of key papers on the subject,
and Mishra~\cite{M04} for a review of techniques to compute with roots of
polynomials.

Collins CAD solves any \emph{geometric} decision problem that does not involve
quantification over the integers in time doubly exponential in the problem
size. This does not harm our results as we exclusively use this algorithm to
solve constant size subproblems. Geometric is to be understood in the sense of Descartes and Fermat, that
is, the geometry of objects that can be expressed with polynomial equations. In
particular, it allows us to make the following computations in the real-RAM and
bounded-degree ADT models:
\begin{enumerate}
\setlength{\itemsep}{0pt}
\setlength{\parskip}{0pt}
\setlength{\parsep}{0pt}
\item Given a constant-degree univariate polynomial, count its real roots
	in $O(1)$ operations,
\item Given a constant number of univariate polynomials of constant degree,
compute the interleaving of their real roots in $O(1)$ operations,
\item Given a point in the plane and an arrangement of a constant number of
constant-degree polynomial planar curves, locate the point in the
arrangement in $O(1)$ operations.
\end{enumerate}

Instead of bounded-degree algebraic decision trees as the nonuniform model
we could consider decision trees in which
each decision involves a constant-size instance of the decision problem in the
first-order theory of the reals. The depth of a bounded-degree algebraic
decision tree simulating such a tree would only be blown up by a constant factor.

\subsection{Previous Results}
\subparagraph{3SUM}
For the sake of simplicity, we consider the following definition of 3SUM
\begin{problem}[3SUM]
Given 3 sets $A$, $B$, and $C$, each containing $n$ real numbers, decide
whether there exist $a \in A$, $b \in B$, and $c \in C$ such that $c=a+b$.
\end{problem}

A quadratic lower bound for solving 3SUM holds in a restricted model of
computation: the $3$-linear decision tree model. Erickson~\cite{E99}
and Ailon and Chazelle~\cite{AC05} showed
that in this model, where one is only allowed to test the sign of a linear
expression of up to three elements of the input, there are a quadratic number of
critical tuples to test.
\begin{theorem}[Erickson~\cite{E99}]
The depth of a \(3\)-linear decision tree for 3SUM is $\Omega(n^2)$.
\end{theorem}

While no evidence suggested that this lower bound could be extended to other
models of computation, it was eventually conjectured that 3SUM requires
$\Omega(n^2)$ time.

Baran et al.~\cite{BDP08} were the first to give concrete evidence
for doubting the conjecture.
They gave subquadratic Las Vegas algorithms for 3SUM, where input
numbers are restricted to be integer or rational, in the circuit RAM,
word RAM, external memory, and cache-oblivious models of computation. Their idea
is to exploit the parallelism of the models, using linear and
universal hashing.

Gr\o nlund~and~Pettie~\cite{GP14}, using a trick due to Fredman~\cite{F76},
recently showed that there exist subquadratic decision trees for 3SUM when the queries
are allowed to be $4$-linear.
\begin{theorem}[Gr\o nlund and Pettie~\cite{GP14}]
There is a $4$-linear decision tree of depth
$O(n^\frac{3}{2} \sqrt{\log n})$ for 3SUM\@.
\end{theorem}
They also gave deterministic and randomized
subquadratic real-RAM algorithms for 3SUM, refuting the conjecture.
Similarly to the subquadratic $4$-linear decision trees, these new results
use the power of $4$-linear queries.
These algorithms were later improved by Freund~\cite{F15} and
Gold~and~Sharir~\cite{GS15}.
\begin{theorem}[Gr\o nlund and Pettie~\cite{GP14}]
There is a
deterministic $O(n^2 {(\log \log n)}^{\sfrac{2}{3}} / {(\log n)}^{\sfrac{2}{3}})$-time
and
a randomized $O(n^2 {(\log \log n)}^2 / \log n )$-time
real-RAM algorithm for 3SUM\@.
\end{theorem}

Since then, the conjecture was eventually updated. This new conjecture is
considered an essential part of the theory of complexity-within-P.
\begin{conjecture}[3SUM Conjecture]\label{conj:3sum}
	There is no $O(n^{2-\delta})$-time algorithm for 3SUM\@.
\end{conjecture}

\subparagraph{Elekes-Rónyai, Elekes-Szabó}
In a series of results spanning fifteen years,
Elekes and Rónyai~\cite{ER00},
Elekes and Szabó~\cite{ES12},
Raz, Sharir and Solymosi~\cite{RSS14}, and
Raz, Sharir and de Zeeuw~\cite{RSZ15}
give upper bounds on the number of solution triples to the 3POL problem.
The last and strongest result is the following
\begin{theorem}[Raz, Sharir and de Zeeuw~\cite{RSZ15}]
	Let $A$, $B$, $C$ be $n$-sets of real numbers and $F \in \mathbb{R}[x,y,z]$
	be a polynomial of constant degree, then the number of triples $(a,b,c) \in
	A \times B \times C$ such that $F(a,b,c)=0$ is $O(n^{11/6})$ unless $F$ has
	some group related form.\footnote{Because our results do not depend on the
	meaning of \emph{group related form}, we do not bother
	defining it here. We refer the reader to Raz, Sharir and de
	Zeeuw~\cite{RSZ15} for the exact definition.}
\end{theorem}

Raz, Sharir and de Zeeuw~\cite{RSZ15} also look at the number of solution
triples for the General Position Testing problem when the input is restricted
to points lying on a constant number of constant-degree algebraic curves.
\begin{theorem}[Raz, Sharir and de Zeeuw~\cite{RSZ15}]\label{thm:rsz15:col}
Let $C_1, C_2, C_3$ be three (not necessarily distinct) irreducible algebraic curves
of degree at most $d$ in $\mathbb{C}^2$, and let $S_1 \subset C_1, S_2 \subset C_2, S_3
\subset C_3$ be finite subsets. Then the
number of proper collinear triples in $S_1 \times S_2 \times S_3$ is
\begin{displaymath}
	O_d( |S_1|^{1/2} |S_2|^{2/3} |S_3|^{2/3} + |S_1|^{1/2} (|S_1|^{1/2} + |S_2| +
|S_3| ) ),
\end{displaymath}
unless $C_1 \cup C_2 \cup C_3$ is a line or a cubic curve.
\end{theorem}

Recently, Nassajian Mojarrad, Pham, Valculescu and de Zeeuw~\cite{MPVd16} and
Raz, Sharir and de Zeeuw~\cite{RSZ16} proved bounds for versions of the
problem where $F$ is a $4$-variate polynomial.

%% file: sec/03-algorithms/01-nonuniform-explicit.tex
\section{Nonuniform algorithm for explicit 3POL}
\label{sec:algo:explicit:nonuniform}

We begin with the description of a nonuniform algorithm for explicit 3POL which
we use later as a basis for other algorithms. We prove the following:
\begin{theorem}\label{thm:explicit:act}
	There is a bounded-degree ADT of depth
	$O(n^{\frac{12}{7}+\varepsilon})$
	for 3POL\@.
\end{theorem}

\subparagraph{Idea}
The idea is to partition the sets $A$ and $B$ into small groups of consecutive
elements. That way, we can divide the $A\times B$ grid into cells with the
guarantee that each curve $c = f(x,y)$ in this grid intersects a small number
of cells. For each such curve and each cell it intersects, we
search $c$ among the values $f(a,b)$ for all $(a,b)$ in a given intersected
cell. We generalize Fredman's trick~\cite{F76} --- and how it is used in Gr\o
nlund and Pettie's paper~\cite{GP14} --- to quickly obtain a sorted order on
those values, which provides us a logarithmic search time for each cell.
Below is a sketch of the algorithm.
\begin{algorithm}[Nonuniform algorithm for explicit 3POL]\label{algo:ne}
\item[input] $A = \{\,a_1 < \cdots < a_n\,\},B = \{\,b_1<\cdots<b_n\,\},
    C = \{\,c_1<\cdots<c_n\,\}
    \subset \mathbb{R}$.
\item[output] \emph{accept} if $\exists\, (a,b,c) \in A \times B \times C$ such that $c
    = f(a,b)$, \emph{reject} otherwise.
\item[1.] Partition the intervals $[a_1,a_n]$ and $[b_1,b_n]$ into blocks
    $A_i^*$ and $B_j^*$ such that $A_i = A \cap A_i^*$ and $B_j = B
    \cap B_j^*$ have size $g$.
\item[2.] Sort the sets $f(A_i \times B_j) = \{\, f(a,b) \st (a,b) \in A_i
    \times B_j\,\}$ for all $A_i,B_j$. This is the only step that
    is nonuniform.
\item[3.] For each $c \in C$,
\item[3.1.] For each cell $A_i^* \times B_j^*$ intersected by the curve
$c=f(x,y)$,
\item[3.1.1.] Binary search for $c$ in the sorted set $f(A_i \times B_j)$.
If $c$ is found, \emph{accept} and halt.
\item[4.] \emph{reject} and halt.
\end{algorithm}
Note that it is easy to modify the algorithm to count or report the solutions.
In the latter case, the algorithm becomes output sensitive.
Like in Gr\o nlund and Pettie's $\tilde{O}(n^{\frac 32})$ decision tree for
3SUM~\cite{GP14}, the tricky part is to give an efficient implementation of
step~\textbf{2}.

\subparagraph{$A \times B$ grid partitioning}
Let $A = \{\,a_1<a_2<\cdots<a_n\,\}$ and $B = \{\,b_1<b_2<\cdots<b_n\,\}$.
For some positive integer $g$ to be determined later, partition the interval
\([a_1,a_n]\) into \(n/g\)
blocks $A_1^*,A_2^*,\ldots,A_{n/g}^*$ such that each block
contains \(g\) numbers in \(A\). Do the same for the interval
\([b_1,b_n]\) with the numbers in \(B\) and name the blocks of this
partition $B_1^*,B_2^*,\ldots,B_{n/g}^*$.
For the sake of simplicity, and without loss of generality, we assume
here that $g$ divides $n$.  We continue to make this assumption in the
following sections.
To each of the \({ ( n/g ) }^2\) pairs of blocks $A_i^*$ and $B_j^*$
corresponds a cell $A_i^* \times B_j^*$.
By definition, each cell contains $g^2$ pairs in \(A \times
B\).
For the sake of notation,
we define
$A_i = A \cap A_i^* = \{\,a_{i,1} < a_{i,2} < \cdots < a_{i,g}\,\}$
and
$B_j = B \cap B_j^* = \{\,b_{j,1} < b_{j,2} < \cdots < b_{j,g}\,\}$.
Figure~\ref{fig:partition} depicts this construction.

\begin{figure}
	\centering
    \begin{subfigure}[t]{0.5\textwidth}
	\includegraphics[width=\textwidth]{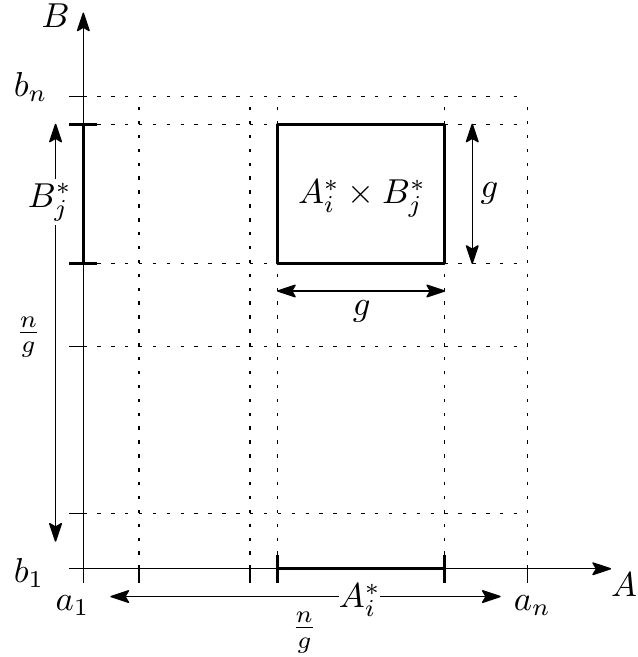}
	\caption{Partitioning \(A\) and \(B\).}\label{fig:partition}
    \end{subfigure}%
    \begin{subfigure}[t]{0.5\textwidth}
	\includegraphics[width=\textwidth]{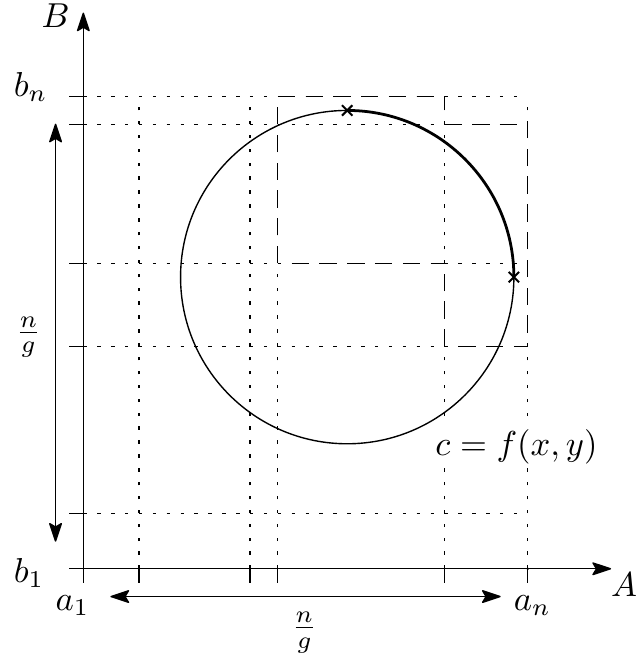}
	\caption{An $xy$-monotone arc of $c=f(x,y)$ intersects a staircase of at most
	$2 \frac ng - 1$ cells in the grid.}\label{fig:intersected-cells}
    \end{subfigure}
    \caption{Properties of the $A\times B$ grid.}\label{fig:axb}
\end{figure}

The following two lemmas result from this construction:
\begin{lemma}\label{lem:intersections}
    For a fixed value \(c \in C\), the curve $c=f(x,y)$ intersects $O(\frac ng)$
    cells. Moreover, those cells can be found in $O(\frac ng)$ time.
\end{lemma}
\begin{proof}
    Since $f$ has constant degree, the curve $c=f(x,y)$ can be decomposed into
    a constant number of $xy$-monotone arcs. Split the curve into
    $x$-monotone pieces, then each $x$-monotone piece into $y$-monotone arcs.
    The endpoints of the $xy$-monotone arcs are the intersections of
    $f(x,y)=c$ with its derivatives $f'_x(x,y)=0$ and $f'_y(x,y)=0$. By
    B\'ezout's theorem, there are $O({\deg(f)}^2)$ such intersections and so
    $O({\deg(f)}^2)$ $xy$-monotone arcs. Figure~\ref{fig:intersected-cells} shows
    that each such arc intersects at most $2 \frac ng - 1$ cells since the
    cells intersected by a $xy$-monotone arc form a staircase in the grid. This proves
    the first part of the lemma.
    To prove the second part, notice that for each connected component of
    $c=f(x,y)$ intersecting at least one cell of the grid either: (1) it
    intersects a boundary cell of the grid, or (2) it is a (singular) point or
    contains vertical and horizontal tangency points. The cells intersected by
    $c=f(x,y)$ are computed by exploring the grid from $O(\frac ng)$ starting cells.
    Start with an empty set. Find and add all boundary cells containing a point of the
    curve. Finding those cells is achieved by solving
    the Tarski sentence
    $\exists (x,y)
    c=f(x,y)
    \land x \in A_i^*
    \land y \in B_j^*$,
    for each cell $A_i^* \times B_j^*$ on the boundary.
    This takes $O(\frac ng)$ time.
    Find and add the cells containing endpoints of $xy$-monotone arcs of
    $c=f(x,y)$. Finding those cells is achieved by first finding
    the constant number of vertical and horizontal slabs $A_i^* \times
    \mathbb{R}$ and $\mathbb{R} \times B_j^*$ containing such points:
    \begin{align*}
	\exists (x,y)
	c=f(x,y)
	\land ( f'_x(x,y)=0 \lor f'_y(x,y)=0)
	\land x \in A_i^*, \\
	\exists (x,y)
	c=f(x,y)
	\land ( f'_x(x,y)=0 \lor f'_y(x,y)=0)
	\land y \in B_j^*.
    \end{align*}
    This takes $O(\frac ng)$ time.
    Then for each pair of vertical and horizontal slab containing such a point,
    check that the cell at the intersection of the slab also contains such a
    point:
    \begin{displaymath}
	\exists (x,y)
	c=f(x,y)
	\land ( f'_x(x,y)=0 \lor f'_y(x,y)=0)
	\land x \in A_i^*
	\land y \in B_j^*.
    \end{displaymath}
    This takes $O(1)$ time.
Note that we can always assume the constant-degree polynomials we manipulate
are square-free, as making them square-free is trivial~\cite{Y76}: since
$\mathbb{R}[x]$ and $\mathbb{R}[y]$ are unique factorization domains, let $Q =
P/\text{gcd}(P,P'_x;x)$ and $\text{sf}(P) = Q/\text{gcd}(P,P'_y;y)$, where
$\text{gcd}(P,Q;z)$ is the greatest common divisor of $P$ and $Q$ when viewed
as polynomials in $R[z]$ where $R$ is a unique factorization domain and
$\text{sf}(P)$
is the square-free part of $P$.
    The set now contains, for each component of each
    type, at least one cell intersected by it. Initialize a list with the
    elements of the set. While the list is not empty, remove any cell from the
    list, add each of the eight neighbouring cells to the set and the list,
    if it contains a point of $c=f(x,y)$ --- this can be checked with the same
    sentences as in the boundary case --- and if
    it is not already in the set. This costs $O(1)$ per cell intersected.
    The set now contains all cells of the grid intersected by $c=f(x,y)$.
\end{proof}

\begin{lemma}\label{lem:preprocessing}
    If the sets $A,B,C$ can be preprocessed in $S_g(n)$ time so that,
    for any given cell $A_i^* \times B_j^*$ and any given $c \in C$, testing whether
    $c \in f(A_i \times B_j) = \{\,f(a,b)\st (a,b) \in A_i \times B_j\,\}$ can be done in
    $O(\log g)$ time, then, explicit 3POL can be solved in
    $S_g(n)+O(\frac{n^2}{g}\log g)$ time.
\end{lemma}
\begin{proof}
    We need $S_g(n)$ preprocessing time plus the time required
    to search each of the $n$ numbers $c \in C$ in each of the $O(\frac ng)$
    cells intersected by $c=f(x,y)$. Each search costs $O(\log g)$ time.
    We can compute the cells intersected by $c=f(x,y)$ in $O(\frac ng)$ time
    by Lemma~\ref{lem:intersections}.
\end{proof}
\remark{} We do not give a $S_g(n)$-time real-RAM
algorithm for preprocessing the input, but only a $S_g(n)$-depth bounded-degree
ADT\@. In fact, this preprocessing step is the only
nonuniform part of Algorithm~\ref{algo:ne}. A real-RAM implementation of this
step is given in
\S\ref{sec:algo:explicit:uniform}.

\begin{figure}
    \centering
    \begin{subfigure}[t]{0.5\textwidth}
    \includegraphics[width=\textwidth]{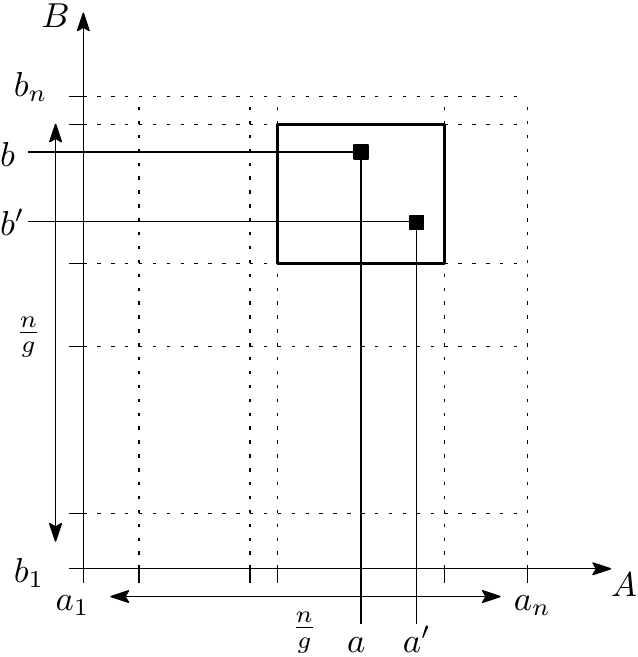}
    \caption{The pairs \((a,b), (a',b')\).}\label{fig:ab-pairs}
    \end{subfigure}%
    \begin{subfigure}[t]{0.5\textwidth}
    \includegraphics[width=\textwidth]{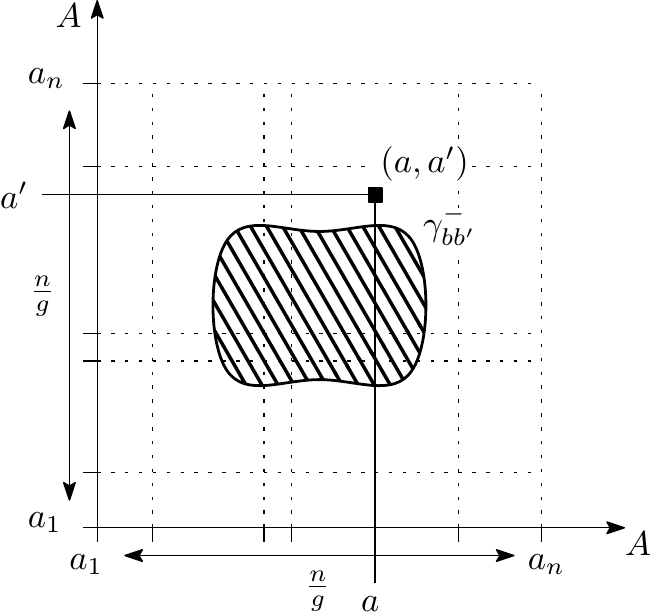}
    \caption{$(a,a') \in \gamma^+_{b,b'}$ implies $f(a,b) > f(a',b')$.}\label{fig:zero-set}
    \end{subfigure}
    \caption{Generalization of Fredman's trick (Lemma~\ref{lem:fredman}).}\label{fig:trick}
\end{figure}
\subparagraph{Preprocessing} All that is left to prove is that $S_g(n)$ is
subquadratic for some choice of $g$. To achieve this we sort the points
inside each cell using Fredman's trick~\cite{F76}. Gr\o nlund and
Pettie~\cite{GP14} use this trick to sort the sets
$A_i + B_j = \{\,a + b \st (a,b) \in A_i \times B_j \,\}$
with few comparisons: sort the set
$D = (\cup_i [A_i - A_i]) \cup (\cup_j [B_j - B_j])$,
where
$A_i - A_i = \{\,a - a' \st (a,a')\in A_i \times A_i \,\}$
and
$B_j - B_j = \{\,b - b' \st (b,b')\in B_j \times B_j \,\}$,
using $O(n\log n + |D|)$ comparisons, then testing whether
$a + b \le a' + b'$
can be done using the free (already computed) comparison
$a - a' \le b' - b$.
We use a generalization of this trick to sort the sets $f(A_i\times B_j)$.
For each $B_j$, for each pair \((b,b') \in B_j \times B_j\),
define the curve
$
	\gamma_{b,b'} = \{\,(x,y) \st f(x,b) = f(y,b')\,\}.
$
Define the sets
$
	\gamma^0_{b,b'} = \gamma_{b,b'},
	\gamma^-_{b,b'} = \{\,(x,y) \st f(x,b) < f(y,b')\,\},
	\gamma^+_{b,b'} = \{\,(x,y) \st f(x,b) > f(y,b')\,\}
$.
The following lemma --- illustrated by Figure~\ref{fig:trick} --- follows by
definition:
\begin{lemma}\label{lem:fredman}
Given a cell $A_i^* \times B_j^*$ and two pairs \((a,b), (a',b') \in A_i \times B_j\),
deciding whether \(f(a,b) < f(a',b')\) (respectively
$f(a,b) = f(a',b')$
and
$f(a,b) > f(a',b')$)
amounts to deciding whether the point
$(a,a')$ is contained in \(\gamma^-_{b,b'}\) (respectively
\(\gamma^0_{b,b'}\) and
\(\gamma^+_{b,b'}\)).
\end{lemma}

There are $N \coloneqq \frac ng \cdot g^2 = ng$ pairs $(a,a') \in \cup_i [A_i \times
A_i]$ and there are $N$ pairs $(b,b') \in \cup_j [B_j \times B_j]$.  Sorting
the $f(A_i\times B_j)$ for all $(A_i, B_j)$ amounts to solving the
following problem:
\begin{problem}[Polynomial Batch Range Searching]
    Given $N$ points and $N$ polynomial curves in $\mathbb{R}^2$, locate
    each point with respect to each curve.
\end{problem}

We can now refine the description of step~\textbf{2} in Algorithm~\ref{algo:ne}
\begin{algorithm}[Sorting the $f(A_i \times B_j)$ with a nonuniform algorithm]\label{algo:sfaixbj}
\item[input] $A = \{\,a_1<a_2<\cdots<a_n\,\},B = \{\,b_1<b_2<\cdots<b_n\,\}
    \subset \mathbb{R}$
\item[output] The sets $f(A_i \times B_j)$, sorted.
\item[2.1.] Locate each point $(a,a') \in \cup_i [A_i \times A_i]$ w.r.t.\
    each curve $\gamma_{b,b'}, (b,b') \in \cup_j [B_j \times
    B_j]$.
\item[2.2.] Sort the sets $f(A_i \times B_j)$ using the
    information retrieved in step \textbf{2.1}.
\end{algorithm}
Note that this algorithm is nonuniform: step~\textbf{2.2} costs at least
quadratic time in the real-RAM model, however, this step does not need to
query the input at all, as all the information needed to sort is retrieved during
step~\textbf{2.1}.
Step~\textbf{2.2} incurs no cost in
our nonuniform model.

To implement step~\textbf{2.1}, we use a modified
version of the $N^{\frac{4}{3}} 2^{O(\log^* N)}$ algorithm of
Matou\v{s}ek~\cite{Ma93} for Hopcroft's problem.
%
%
%
%
In Appendix~\ref{sec:algo:point-curves-location},
we prove the following upper bound:
\begin{lemma}\label{lem:batch}
Polynomial Batch Range Searching can be solved in
$O(N^{\frac{4}{3}+\varepsilon})$ time in the real-RAM model when the input
curves are the $\gamma_{b,b'}$.
\end{lemma}


\subparagraph{Analysis}
Combining Lemma~\ref{lem:preprocessing} and Lemma~\ref{lem:batch} yields a $O(
{(ng)}^{4/3+\varepsilon} + n^2 g^{-1} \log g)$-depth bounded-degree ADT for
3POL\@. By optimizing over $g$, we get $g = \Theta(n^{2/7-\varepsilon})$, and
the previous expression simplifies to $O(n^{12/7+\varepsilon})$, proving
Theorem~\ref{thm:explicit:act}.

%% file: sec/03-algorithms/03-uniform-explicit.tex
\section{Uniform algorithm for explicit 3POL}
\label{sec:algo:explicit:uniform}

We now build on the first algorithm and prove the following:
\begin{theorem}\label{thm:explicit:uniform}
	Explicit 3POL can be solved in
	$O(n^2 {(\log \log n)}^\frac{3}{2} / {(\log n)}^\frac{1}{2})$ time.
\end{theorem}
We generalize again Gr\o nlund and Pettie~\cite{GP14}. The
algorithm we present is derived from the first subquadratic algorithm in their
paper.

\subparagraph{Idea}
We want the implementation of step~\textbf{2} in Algorithm~\ref{algo:ne} to be
uniform, because then, the whole algorithm is.
We use the same partitioning scheme as before except we choose $g$ to be much
smaller.
This allows to
store all permutations on $g^2$ items in a lookup table, where
$g$ is chosen small enough to make the size of the lookup table $\Theta(n^\varepsilon)$.
The preprocessing part of the previous algorithm is replaced by $ g^2! $ calls to
an algorithm that determines for which cells a given permutation gives the
correct sorted
order. This preprocessing step stores a constant-size%
\footnote{In the real-RAM and word-RAM models.}
pointer from each cell to
the corresponding permutation in the lookup table.
Search can now be done efficiently: when searching a value $c$ in $f(A_i
\times B_j)$, retrieve the corresponding permutation on $g^2$ items from the
lookup table,
then perform binary search on the sorted order defined by that permutation.
The sketch of the algorithm is exactly Algorithm~\ref{algo:ne}. The only
differences with respect to \S\ref{sec:algo:explicit:nonuniform} are the choice
of $g$ and the implementation of step~\textbf{2}.

\subparagraph{$A \times B$ grid partitioning}
We use the same partitioning scheme as before, hence
Lemma~\ref{lem:intersections} and Lemma~\ref{lem:preprocessing} hold.
We just need to find a replacement for Lemma~\ref{lem:batch}.

\subparagraph{Preprocessing}
For their simple subquadratic 3SUM algorithm, Gr\o nlund and Pettie~\cite{GP14}
explain that for a permutation to give the correct sorted order for a cell, that
permutation must define a \emph{certificate} --- a set of inequalities --- that
the cell must verify. They cleverly note --- using Fredman's Trick~\cite{F76}
as in Chan~\cite{Cha08} and Bremner et al.~\cite{BCDEHILPT14}
--- that the verification of a single certificate by all cells amounts to
solving a red/blue point dominance reporting problem.
We generalize their method.
For each permutation $\pi\colon\,[g^2]\to{[g]}^2$, where $\pi =
(\pi_r,\pi_c)$ is decomposed into row and column functions
$\pi_r,\pi_c\colon\,[g^2]\to[g]$, we enumerate all cells $A_i^* \times B_j^*$
for which the following \emph{certificate} holds:
\begin{displaymath}
	f(a_{i,\pi_r(1)},b_{j,\pi_c(1)}) \le f(a_{i,\pi_r(2)},b_{j,\pi_c(2)}) \le \cdots
	\le f(a_{i,\pi_r(g^2)},b_{j,\pi_c(g^2)}).
\end{displaymath}

\remark{}
Since some entries may be equal, to make sure each cell corresponds to exactly
one certificate, we replace $\le$ symbols by choices of $g^2-1$ symbols in
$\{\,=,<\,\}$. Each permutation $\pi$ gets a certificate for each of
those choices.
This adds a $2^{g^2-1}$ factor to the number of certificates
to test, which will eventually be negligible.
Note that some of those $2^{g^2-1}$ certificates are equivalent. We
need to skip some of them, as otherwise we might output some cells more than
once, and then there will be no guarantee with respect to the output size. For
example, the certificate $f(a_{i,9},b_{j,5}) = f(a_{i,6},b_{j,7}) < \cdots <
f(a_{i,4},b_{j,4})$ is equivalent to the certificate $f(a_{i,6},b_{j,7}) =
f(a_{i,9},b_{j,5}) < \cdots < f(a_{i,4},b_{j,4})$. Among equivalent
certificates, we only consider the certificate whose permutation $\pi$ precedes
the others lexicographically. In the previous example,
$((6,7),(9,5),\ldots,(4,4)) \prec ((9,5),(6,7),\ldots,(4,4))$ hence we would
only process the second certificate. For the sake of simplicity,
we will write inequality when we mean strict inequality or equation, and
``$\le$'' when we mean ``$<$'' or ``$=$''.

\subparagraph{Fredman's Trick}
This is where Fredman's Trick comes into play.
By Lemma~\ref{lem:fredman}, each inequality
$f(a_{i,\pi_r(t)},b_{j,\pi_c(t)}) \le f(a_{i,\pi_r(t+1)},b_{j,\pi_c(t+1)})$
of a certificate can be checked
by computing the relative position of $(a_{i,\pi_r(t)},a_{i,\pi_r(t+1)})$ with
respect to $\gamma_{b_{j,\pi_c(t)},b_{j,\pi_c(t+1)}}$.
For a given certificate, for each $A_i$ and each $B_j$, define
\begin{align*}
p_i &= (
	(a_{i,\pi_r(1)},a_{i,\pi_r(2)}),
	(a_{i,\pi_r(2)},a_{i,\pi_r(3)}),
	\ldots,
	(a_{i,\pi_r(g^2-1)},a_{i,\pi_r(g^2)})
),\\
q_j &= \mleft(
f(x,b_{j,\pi_c(1)}) \le f(y,b_{j,\pi_c(2)}),
\ldots,
f(x,b_{j,\pi_c(g^2-1)}) \le f(y,b_{j,\pi_c(g^2)})
\mright).
\end{align*}
A certificate is verified by a cell $A_i \times B_j$ if and only if, for all
$t \in [g^2-1]$, the point $p_{i,t}$ verifies the inequality $q_{j,t}$.
Enumerating all cells $A_i\times B_j$ for which the
certificate holds therefore amounts to solving the following problem:
\begin{problem}[Polynomial Dominance Reporting (PDR)]
Given $N$ $k$-tuples $p_i$ of points in $\mathbb{R}^2$ and $N$ $k$-tuples $q_j$
of bivariate polynomial inequalities of degree at most $\deg(f)$,
enumerate all pairs $(p_i,q_j)$ where, for all $t \in [k]$,
the point $p_{i,t}$ verifies the inequality $q_{j,t}$.
\end{problem}
In the next section, we explain how to solve PDR efficiently
and prove the following lemma:
\begin{lemma}\label{lem:dominance}
    We can enumerate
    all $\ell$ such pairs in time
    $2^{O(k)} N^{2-\frac{4}{{\deg(f)}^2+3\deg(f)+2}+\varepsilon} + O(\ell)$.
\end{lemma}

We can now give a uniform implementation of step~\textbf{2} in
Algorithm~\ref{algo:ne}:
\begin{algorithm}[Sorting the $f(A_i \times B_j)$ with a uniform algorithm]\label{algo:sfaixbj}
\item[input] $A = \{\,a_1<a_2<\cdots<a_n\,\},B = \{\,b_1<b_2<\cdots<b_n\,\}
    \subset \mathbb{R}$
\item[output] The sets $f(A_i \times B_j)$, sorted.
\item[2.1.] Initialize a lookup table that will contain all $O(2^{g^2-1}(g^2!))$ certificates on $g^2$ elements.
\item[2.2.] For each permutation $\pi\colon\,[g^2]\to{[g]}^2$,
\item[2.2.1.] For each choice of $g^2-1$ symbols in $\{\,=,<\,\}$,
\item[2.2.1.1.] If there is any ``$=$'' symbol that corresponds to a
    lexicographically decreasing
    pair of tuples of indices in $\pi$, skip this choice of symbols.
\item[2.2.1.2.] Append the certificate associated to $\Pi$
    and the choice of symbols to the table.
\item[2.2.1.3.] Solve the PDR instance associated
    to $A,B,\Pi$ and the choice of symbols.
\item[2.2.1.4.] For each output pair $(i,j)$, store a pointer
    from $(i,j)$ to the last entry in the table.
\end{algorithm}

\subparagraph{Analysis} Plugging in $k=g^2-1$, $N=\frac ng$, iterating over all
permutations ($\sum_{\pi} \ell = {(n/g)}^2$), and adding the binary search step we
get that
explicit 3POL can be solved in time
\begin{displaymath}
    (g^2!)2^{g^2}2^{O(g^2)}
    {(n/g)}^{2-\frac{4}{{\deg(f)}^2+3\deg(f)+2}+\varepsilon} + O({(n/g)}^2) +
    O(n^2 \log g / g).
\end{displaymath}
The first two terms correspond to the complexity of step~\textbf{2} in
Algorithm~\ref{algo:ne}, and the last term corresponds to the complexity of
step~\textbf{3} in Algorithm~\ref{algo:ne}.
To get subquadratic time we can set
$g = c_{\deg(f)}\sqrt{\log n/\log \log n}$,
because then
for some appropriate choice of the constant factor $c_{\deg(f)}$,
$(g^2)!2^{g^2}2^{O(g^2)} = n^{\delta}$ where $\delta<
4/({\deg(f)}^2+3\deg(f)+2) - \varepsilon$, making the first term negligible.
The complexity of the algorithm is dominated by $O(n^2 \log g / g) = O(n^2
{(\log \log n)}^{\frac 32} / {(\log n)}^{\frac 12} )$.
This proves Theorem~\ref{thm:explicit:uniform}.


%% file: sec/03-algorithms/04-dominance-reporting.tex
\section{Polynomial Dominance Reporting}
\label{sec:algo:dominance}


In this section, we combine a standard dominance reporting
algorithm~\cite{PS85} with Matou\v{s}ek's algorithm~\cite{Ma93} to prove
Lemma~\ref{lem:dominance}.
We say a pair of blue and red points in $\mathbb{R}^k$ is dominating if for all
indices $i\in[k]$ the $i$th coordinate of the blue point is greater or
equal to the $i$th coordinate of the red point.
The standard algorithm~\cite{PS85} solves the following problem:
\begin{problem}
	Given $N$ blue
	and $M$ red points in $\mathbb{R}^k$, report all
	bichromatic dominating pairs.
\end{problem}
Our problem is significantly more complicated and general. Instead of blue points we
have blue $k$-tuples $p_i$ of $2$-dimensional points, instead of red points we have
red $k$-tuples $q_j$ of bivariate polynomial inequalities, and we want to report
all bichromatic pairs $(p_i,q_j)$ such that, for all $t \in [k]$,
the point $p_{i,t}$ verifies the inequality $q_{j,t}$.
The standard algorithm essentially works by a combination of divide and conquer
and prune and search, using a one-dimensional cutting (median selection) to
split a problem into subproblems. We generalize the standard algorithm by using
higher dimensional cuttings, in a way similar to Matou\v{s}ek's
algorithm~\cite{Ma93}. For the analysis, we generalize Chan's analysis
of the standard algorithm when $k$ is not constant~\cite{Cha08}.

\begin{proof}[Proof of Lemma~\ref{lem:dominance}]

We use the Veronese embedding~\cite{Har77,Har13}.
Since the polynomials have constant degree, we can trade polynomial
inequalities for linear inequalities by lifting everything to a space of
higher --- but constant --- dimension. The degree of each polynomial is at most
$\deg(f)$. There are exactly $d = \binom{\deg(f)+2}{2} - 1$ different bivariate
monomials of degree at most $\deg(f)$%
\footnote{Not including the independent monomial, namely, $1$.}.
To each monomial we associate a variable
in $\mathbb{R}^d$. By this association, points in the plane are mapped to
points in $\mathbb{R}^d$ and bivariate polynomial inequalities are mapped to
$d$-variate linear inequalities.

By abuse of notation, let $p_i$ denote the tuple $p_i$ where each
$2$-dimensional point has been replaced by its $d$-dimensional counterpart, and
let $q_i$ denote the tuple $q_i$ where each bivariate polynomial inequality
has been replaced by its $d$-variate linear counterpart.
We have $N$ $k$-tuples $p_i$ and $M$ $k$-tuples $q_j$.
The algorithm checks each of the $k$ components of the tuples in turn and
can be described recursively as follows for some positive integer $r > 1$:
\begin{algorithm}[Polynomial Dominance Reporting]\label{algo:pdr}
\item[input] $N$ $k$-tuples $p_i$ of $d$-dimensional points, $M$ $k$-tuples
	$q_j$ of $d$-variate linear inequalities.
\item[output] All $(p_i,q_j)$ pairs such that, for all $t \in [k]$, the point
	$p_{i,t}$ verifies the inequality $q_{j,t}$.
\item[1.] If $k=0$, then output all pairs $(p_i,q_j)$ and halt.
\item[2.] If $N < r^d$ or $M < r$, solve the problem by brute force
	in $O((N+M) k)$ time.
\item[3.]
We now only consider the $k$th component of each input $k$-tuple and call these
\emph{active} components.
To each active $d$-variate linear inequality corresponds a defining
hyperplane in $\mathbb{R}^d$.
Construct, as in~\cite{Ma93}, a hierarchical cutting of $\mathbb{R}^d$ using $O(r^d)$
simplicial cells such that each simplicial cell
is intersected by at most $\frac{M}{r}$ of the defining hyperplanes.
This construction also gives us for each simplical cell of the cutting the list
of defining hyperplanes intersecting it.
This takes $O(Mr^{d-1})$ time.
Locate each active point inside the hierarchical cutting in time $O(N \log r)$.
Let $S$ be a simplicial cell of the hierarchical cutting.
Denote by $\Pi_S$ the set of active points in $S$. Partition each $\Pi_S$ into
$\left\lceil \frac{\lvert \Pi_S \rvert}{N r^{-2}} \right\rceil$ disjoint subsets
of size at most $\frac{N}{r^d}$.
For each simplicial cell, find the active inequalities whose corresponding
geometric object (hyperplane, closed or open half-space) contains the cell.
This takes $O(Mr^d)$ time.
The whole step takes $O(N\log r+Mr^d)$ time.
\item[4.] For each of the $O(r^d)$ simplicial cells, recurse on the at most
	$\frac{N}{r^d}$ $k$-tuples $p_i$ whose active point is inside the
	simplicial cell and the at most $\frac{M}{r}$ $k$-tuples $q_j$ whose active
	inequality's defining hyperplane intersects the simplicial cell.
\item[5.] For each of the $O(r^d)$ simplicial cells, recurse on the at most
	$\frac{N}{r^d}$ ($k-1$)-prefixes of $k$-tuples $p_i$ whose active point is
	inside the simplicial cell and the ($k-1$)-prefixes of $k$-tuples $q_j$
	whose active inequality's corresponding geometric object contains the simplicial cell.
\end{algorithm}

\subparagraph{Correctness}
In each recursive call, either $k$ is decremented or $M$ and $N$ are divided by
some constant, hence, one of the conditions in steps \textbf{1} and \textbf{2} is
met in each of the paths of the recursion tree and the algorithm always terminates.
Step \textbf{5} is correct because it only recurses on $(p_i,q_j)$ pairs whose
suffix pairs are dominating.
The base case in step \textbf{1} is correct because the only way for a pair
$(p_i,q_j)$ to reach this point is to have had all $k$ components checked in
step \textbf{5}.
The base case in step \textbf{2} is correct by definition.
Each dominating pair is output exactly once because the recursive calls of
step \textbf{4} and \textbf{5} partition the set of pairs $(p_i,q_j)$ that can
still claim to be candidate dominating pairs.

\subparagraph{Analysis} For $k,N,M\ge0$, the total complexity $T_k(N,M)$ of
computing the inclusions for the first $k$ components, excluding the output
cost (steps~\textbf{1}~and~\textbf{2}), is bounded by
\begin{align*}
	T_k(N,M)
	&\le
	\underbrace{O(r^d)\,T_{k-1}(N,M)}_{\text{Step}~\textbf{5}}
	+
	\underbrace{O(r^d)\,T_k\mleft(\frac{N}{r^d},\frac{M}{r}\mright)}_{\text{Step}~\textbf{4}}
	+
	\underbrace{O(N+M)}_{\text{Step}~\textbf{3}},\\
	T_0(N,M) &= 0, 
	\,\,\,\,\,T_k(N,M) = O(Nk)~\text{if $M < r$}, 
	\,\,\,\,\,T_k(N,M) = O(Mk)~\text{if $N < r^d$}.
\end{align*}
By point-hyperplane duality, $T_k(N,M) = T_k(M,N)$, hence, we can execute
step \textbf{4} on dual linear inequalities and dual points to balance the
recurrence. For some constant $c_1 \ge 1$,
\begin{displaymath}
	T_k(N,M)
	\le
	c_1 r^{2d} T_{k-1}(N,M)
	+
	c_1 r^{2d} T_k\mleft(\frac{N}{r^{d+1}},\frac{M}{r^{d+1}}\mright)
	+
	c_1 (N+M).
\end{displaymath}
For simplicity, we ignore some problem-size
reductions occuring in this balancing step.

Let $T_k(N) = T_k(N,N)$ denote the complexity of solving the problem when
$M=N$, excluding the output cost. Hence,
\begin{align*}
	T_k(N)
	&\le
	c_1 r^{2d} T_{k-1}(N)
	+
	c_1 r^{2d} T_k\mleft(\frac{N}{r^{d+1}}\mright)
	+
	c_1 N \addtocounter{equation}{1}\tag{\theequation} \label{eq:tkn},\\
	T_0(N) &= 0, 
	\,\,\,\,\,T_k(N) = O(k)~\text{if $N < r^{d+1}$}.
\end{align*}

Solving the recurrence\footnote{See Appendix~\ref{app:apdr}.} gives
$T_k(N) = 2^{O(k)}N^{\frac{2d}{d+1}+\varepsilon_r}$,
and since $d=\binom{\deg(f)+2}{2}-1$, we have
\begin{displaymath}
	T_k(N) = 2^{O(k)}N^{2-\frac{4}{{\deg(f)}^2+3\deg(f)+2}+\varepsilon_r}.
\end{displaymath}
To that complexity we add a constant time unit for each output
pair in steps \textbf{1} and \textbf{2}.
\end{proof}


%% file: sec/03-algorithms/99-3pol.tex
\section{3POL}

Extending the previous techniques
to work for the (implicit) 3POL problem is nontrivial:
\begin{enumerate}
\setlength{\itemsep}{0pt}
\setlength{\parskip}{0pt}
\setlength{\parsep}{0pt}
\item Instead of sorting the sets $f(A_i \times B_j)$ we need to sort the
real roots of the
$F(A_i \times B_j, z)$,
\item The $\gamma_{b,b'}$ curves must be redefined. The redefined curve
$\gamma_{b,b'}$ is still the zero-set of some constant-degree bivariate
polynomial $P(x,y)$. However, retrieving the information we need for sorting
becomes more challenging than just computing the sign of the $P(A_i \times A_i)$,
\item The implementation of the certificates for the uniform algorithm gets
much more convoluted: each certificate checks
the validity of a conjunction of Tarski sentences.
\end{enumerate}

Those extensions are explained in detail in
Appendix~\ref{sec:algo:implicit:nonuniform}~and~\ref{sec:algo:implicit:uniform}
where we show
\begin{theorem}\label{thm:act-3POL-bis}
	There is a bounded-degree ADT of depth
	$O(n^{\frac{12}{7}+\varepsilon})$ for 3POL\@.
\end{theorem}
\begin{theorem}\label{thm:implicit:uniform-bis}
	3POL can be solved in
	$O(n^2 {(\log \log n)}^\frac{3}{2} / {(\log n)}^\frac{1}{2})$ time.
\end{theorem}

%% file: sec/03-algorithms/02-batch-range-searching.tex
\section{Polynomial Batch Range Searching}
\label{sec:algo:point-curves-location}

In this section we present a uniform algorithm that computes the relative
position of \(M\) points with respect to \(N\) $\gamma_{b,b'}$ curves. We call
such a problem an \((M,N)\)-problem. When \(M=N\) the complexity of the
algorithm is \(O(N^{\frac{4}{3}+\varepsilon})\).
The algorithm gives the output in ``concise form'':
it outputs a set of $(\Pi_\alpha, \Gamma_\beta, \sigma)$ triples where
$\Pi_\alpha$ is a subset of input points, $\Gamma_\beta$ is a subset of input curves, and $\sigma \in
\{\,-,0,+\,\}$ indicates the relative position of all points in $\Pi_\alpha$ with
respect to all curves in $\Gamma_\beta$.
Note that if one is only interested in incident point-curve pairs, the
algorithm can explicitely report all of them in
$O(N^{\frac{4}{3}+\varepsilon})$ time, because there are at most $O(N^{\frac
43})$ such pairs and because they can easily be filtered from the output.

\subparagraph{Tools} The proof of Lemma~\ref{lem:batch} involves stantard
computational geometry tools: vertical decomposition of an
arrangement of polynomial curves, $\varepsilon$-nets, cuttings and
derandomization. For the construction of the vertical
decomposition of an arrangment of polynomial curves, we refer the reader to
Pach and Sharir~\cite{Alcala}, Chazelle et al.~\cite{CEGS91}, and Edelsbrunner
et al.~\cite{EGPPSS92}. For cuttings, $\varepsilon$-nets and derandomization, we
refer the reader to Matou\v{s}ek~\cite{M95,M96}, Chazelle and
Matou\v{s}ek~\cite{CM96} and Brönnimann et al.~\cite{BCM99}.

\begin{proof}[Proof of Lemma~\ref{lem:batch}]

Fix some constant $r > 1$.
If $M < r^2$ or $N < r$, solve by brute-force in $O(M+N)$ time. Otherwise,
consider the range space defined by $\gamma_{b,b'}$ curves and $y$-axis aligned
trapezoidal patches whose top and bottom sides are pieces of $\gamma_{b,b'}$
curves. This range space has constant VC-dimension.
Compute an $\frac 1r$-net of size $O(r \log r)$ for the input curves with respect to
this range space. Compute the vertical decomposition \(\Xi\) of the arrangement
of this $\frac 1r$-net. This decomposition is a $\frac 1r$-cutting: it partitions
$\mathbb{R}^2$ into $O(r^2 \log^2 r)$ cells of constant complexity each of
which intersects at most \(\frac{N}{r}\) input curves. Denote by \(\Pi_C\) the
set of points contained in the cell \(C \in \Xi\). Partition each \(\Pi_C\)
into \(\left\lceil \frac{\lvert \Pi_C \rvert}{M r^{-2}} \right\rceil\) disjoint subsets of
size at most \(\frac{M}{r^2}\). All of this can be done in \(O(M+N)\) time.
The last step consists of solving \(O(r^2 \log^2 r)\)
\((\frac{M}{r^2},\frac{N}{r})\)-problems, that is, solving the problem
recursively for the points and curves intersecting each cell.
The recursive call will be done by swapping the role of the points and curves
using a form of duality to be described below.
%

\subparagraph{Correctness} We want to locate each point with respect to each
curve. When considering a curve-cell pair, there are two cases: (1) either the
curve intersects the cell, or (2) it does not. For the first case we locate
each point in the cell with respect to the curve in one of the recursive steps.
For the second step, the relative position of all points in the cell with
respect to the curve is the same, it suffices thus to locate one of those point
with respect to the curve to get the location of all the points in
$O(1)$ time.
Each recursive call divides $M$ and
$N$ by some constant, hence, the base case is reached
in each of the paths of the recursion tree and the algorithm always terminates.

\subparagraph{Analysis} For $c_1$ some constant and bounding $c_1 r^{2} \log^2 r$
above by $c_2 r^{2+\varepsilon}$ for some large enough constant $c_2$, the
complexity $T(M,N)$ of an $(M,N)$-problem is thus
\begin{align*}
    T(M,N) &\le c_2(r^{2+\varepsilon})\,T\mleft(\frac{M}{r^{2}},\frac{N}{r}\mright) + O(M+N).
\end{align*}
The complexity $T(N,M)$ of a \((N,M)\)-problem is the same as the
complexity $T(M,N)$ of an \((M,N)\)-problem by the following point-curve
duality result whose proof is straightforward
\begin{lemma}\label{lem:dual}
Define
\begin{displaymath}
    \hat{\gamma}_{a,a'} = \{\, (x,y) \st f(a,x)=f(a',y)\,\},
\end{displaymath}
then, locating $(a,a')$ with respect to $\gamma_{b,b'}$ amounts to
locating $(b,b')$ with respect to $\hat{\gamma}_{a,a'}$.
\end{lemma}

By doing alternately one step in the primal with the points $(a,a')$ and the
curves $\gamma_{b,b'}$, then a second step with the dual points $(b,b')$ and the
dual curves $\hat{\gamma}_{a,a'}$, we get the following recurrence
\begin{align*}
    T(M,N) &\le c_2^2(r^{4+\varepsilon})\,T\mleft(\frac{M}{r^{3}},\frac{N}{r^{3}}\mright) +
	c_2(r^{2+\varepsilon})\,O\mleft(\frac{M}{r^2} + \frac{N}{r}\mright) +
	O(M+N)\\
	&\le c_2^2(r^{4+\varepsilon})\,T\mleft(\frac{M}{r^{3}},\frac{N}{r^{3}}\mright) +
	O(M+N)
\end{align*}
Hence, for some large enough constant $c_3$ (using the Master Theorem),
\begin{align*}
T(N,N) = T(N) &\le c_3 (r^{4+\varepsilon})\,T\mleft(\frac{N}{r^{3}}\mright) + O(N)\\
	&\le O(N^{\log_{r^{3}} c_3  r^{4+\varepsilon}})\\
	&\le O(N^{\frac{4}{3}+\varepsilon}).
\end{align*}
\end{proof}

Let us recapitulate the whole algorithm,
\begin{algorithm}[Polynomial Batch Range Searching]\label{algo:pbrs}
\item[input] A set $\Pi$ of $M$ points $(a,a')$, A set $\Gamma$ of $N$ curves $\gamma_{b,b'}$.
\item[output] A set of triples $(\Pi_\alpha,\Gamma_\beta,\sigma)$ covering $\Pi \times \Gamma$
    such that for any triple $(\Pi_\alpha,\Gamma_\beta,\sigma)$,
    for all point $(a,a')$ in $\Pi_\alpha$ and all curve $\gamma$ in
    $\Gamma_\beta$,
    $(a,a') \in \gamma^\sigma$.
\item[0.] If $M < r^2$ or $N < r$, solve the problem by brute force in $O(M+N)$ time.
\item[1.] Compute an $\frac 1r$-net of size $O(r \log r)$ for the input curves.
\item[2.] Compute the vertical decomposition \(\Xi\) of the arrangement of this
    $\frac 1r$-net.
\item[3.] Denote by \(\Pi_C\) the set of points contained in the cell \(C \in \Xi\).
    Partition each \(\Pi_C\) into \(\left\lceil \frac{\lvert \Pi_C \rvert}{M
    r^{-2}} \right\rceil\) disjoint subsets
    $\Pi_{C,i}$ of size at most \(\frac{M}{r^2}\).
\item[4.] For each cell $C$ of the vertical decomposition,
\item[4.1.] For each subset $\Pi_{C,i}$ of points contained in that cell,
\item[4.1.1.] Solve an $(\frac{N}{r},\frac{M}{r^2})$-problem on the curves
    intersecting that cell and the points in $\Pi_{C,i}$, swapping the roles of
    lines and curves via duality.
\item[4.2.] For each curve $\gamma$ not intersecting $C$,
\item[4.2.1.] Compute the location $\sigma_{C,\gamma}$ of any point in $C$ with
    respect to $\gamma$.
\item[4.3.] Output $(\{\,\gamma \st \sigma_{C,\gamma} = -\,\},\Pi_C,
    -)$.\footnote{Note that $\Pi_C$ is implemented as a pointer to the input
    points in $C$.}
\item[4.4.] Output $(\{\,\gamma \st \sigma_{C,\gamma} = +\,\},\Pi_C, +)$.
\item[4.5.] Output $(\{\,\gamma \st \sigma_{C,\gamma} = 0\,\},\Pi_C,
    0)$.\footnote{Some cells of the vertical decomposition could be degenerate trapezoidal
    patches. The decomposition could contain vertices, line segments, and curve
    segments as cells, each of which could contain input points and be contained
    by an input curve.}
\end{algorithm}




%% file: 03-algorithms-04-dominance-reporting-analysis.tex
\section{Analysis of Polynomial Dominance Reporting}\label{app:apdr}
To get rid of the parameter $k$ and progress into the analysis of the
recurrence, Chan makes an ingenious change of variable~\cite{Cha08}.
With hindsight,
choose $b = r^{d+1}$ and let
\begin{equation}\label{eq:tn'}
	T(N') = \max_{b^k N\le N'} T_k(N),
\end{equation}
where the maximum is taken over all integers $k\ge0,N\ge1$.
By combining~(\ref{eq:tkn})~and~(\ref{eq:tn'}) we obtain
\begin{displaymath}
	T(N') = \max_{b^k N\le N'} T_k(N)
	\le
	\max_{b^k N\le N'} \mleft[ c_1r^{2d} T_{k-1}\mleft(N\mright) + c_1r^{2d}
	T_k\mleft(\frac{N}{r^{d+1}}\mright) + c_1 N \mright].
\end{displaymath}
The maximum of a sum is always bounded by the sum of the maxima of its terms,
hence,
\begin{displaymath}
	T(N') \le
	\max_{b^k N\le N'} \mleft[ c_1 r^{2d} T_{k-1} \mleft(N\mright)\mright]
	+
	\max_{b^k N\le N'} \mleft[ c_1 r^{2d} T_k \mleft(\frac{N}{r^{d+1}}\mright)\mright]
	+
	\max_{b^k N\le N'} \mleft[ c_1 N \mright].
\end{displaymath}
By definition of $T(N')$, we have
\begin{align*}
\max_{b^k N\le N'} T_{k-1}(N)
&=
\max_{b^{k-1} N\le \frac{N'}{b}} T_{k-1}(N)
=
T\mleft(\frac{N'}{b}\mright)
=
T\mleft(\frac{N'}{r^{d+1}}\mright),\\
\max_{b^k N\le N'} T_{k}\mleft(\frac{N}{r^{d+1}}\mright)
&=
\max_{b^k \frac{N}{r^{d+1}}\le \frac{N'}{r^{d+1}}} T_{k}\mleft(\frac{N}{r^{d+1}}\mright)
=
T\mleft(\frac{N'}{r^{d+1}}\mright),\\
\max_{b^k N\le N'} N
&=
\max_{N\le \frac{N'}{b^k}} N
=
\frac{N'}{b^k} \le N',
\end{align*}
which, when combined, produce the following recurrence
\begin{displaymath}
	T(N') \le 2 c_1 r^{2d} T\mleft(\frac{N'}{r^{d+1}}\mright) + c_1 N'.
\end{displaymath}

\subparagraph{Powers of $r^{d+1}$}
We claim that if $N'$ is a power of $r^{d+1}$, then $T(N') \le c_2[{N'}^\alpha - N']$
for some constants $\alpha > 1$ and $c_2 \ge 1$. We prove by
induction that this guess is indeed correct.
For $N' = 1$, we have
\begin{align*}
	T(1)
	=
	\max_{b^{k}N \le 1} T_{k} (N)
	=
	T_0(1)
	=
	0
	\le
	c_2[1^{\alpha} - 1].
\end{align*}
For $N' \ge r^{d+1}$ a power of $r^{d+1}$,
assuming the claim holds for all smaller powers of $r^{d+1}$
\begin{align*}
	T(N')
	&\le
	2 c_1r^{2d}c_2\mleft[\mleft(\frac{N'}{r^{d+1}}\mright)^\alpha
	-\frac{N'}{r^{d+1}}\mright]
	+
	c_1 N'\\
	&\le c_2 {N'}^\alpha
	\mleft[\frac{2c_1r^{2d}}{(r^{d+1})^\alpha}\mright]
	- c_2 N' \mleft[2c_1r^{d-1} -
	\frac{c_1}{c_2}\mright]
\end{align*}
We want
\begin{equation*}
\frac{c_1r^{2d}}{(r^{d+1})^\alpha} \le \frac 12
\qquad
\text{and}
\qquad
2c_1r^{d-1} - \frac{c_1}{c_2} \ge 1.
\end{equation*}
For the first inequality, we can set the left hand side to be equal to
$c_1 r^{-\varepsilon'} = \frac{1}{2}$
with some small $\varepsilon' = \frac{1 + \log c_1}{\log r}$.
Hence,
$2d - \alpha(d+1) = -\varepsilon'$, and for $\varepsilon = \frac
{\varepsilon'} {d+1}$, we get $\alpha = \frac{2d}{d+1} + \varepsilon$.
The second inequality is equivalent to
$2r^{d-1} \ge \frac{1}{c_1} + \frac{1}{c_2}$,
which always holds since
$r \ge 2, d \ge 1, c_1 \ge 1, c_2 \ge 1$.

We now have
\begin{displaymath}
	T(N') = O({N'}^{\frac{2d}{d+1}+\varepsilon}),
\end{displaymath}
where
$\varepsilon = \frac{1+\log c_1}{(d+1)\log r}$
can be chosen arbitrarily small by picking
$r = (2c_1)^{\sfrac{1}{\varepsilon (d+1)}}$
arbitrarily large.

\remark{}
The choice $b=r^{d+1}$ gives a simpler analysis. Although giving more
freedom to the value of $b$ --- as in Chan's paper --- yields a slightly better
relation between $\varepsilon$ and $r$, namely
$r>c_1^{\sfrac{1}{\varepsilon (d+1)}}$, it does not get rid of the dependency
of $\varepsilon$ in $r$, unless $c_1=1$.

\subparagraph{General case}
When $N' \ge 2$ is not a power of $r^{d+1}$, we use the fact that $T(N') \le
T(N'+1)$ by definition,
\begin{align*}
	T(N')
	&=
	T\mleft({(r^{d+1})}^{\log_{r^{d+1}} N'}\mright)
	\le
	T\mleft({(r^{d+1})}^{\lfloor\log_{r^{d+1}} N'\rfloor + 1}\mright)
	=
	O\mleft({(r^{d+1})}^{(\lfloor\log_{r^{d+1}} N'\rfloor + 1
	)(\frac{2d}{d+1}+\varepsilon)}\mright)
	\\
	&=
	O\mleft({(r^{d+1})}^{\frac{2d}{d+1}+\varepsilon}
	{{(r^{d+1})}^{\lfloor\log_{r^{d+1}} N'\rfloor}}^{\frac{2d}{d+1}+\varepsilon}\mright)
	=
	O\mleft({{(r^{d+1})}^{\lfloor\log_{r^{d+1}}
	N'\rfloor}}^{\frac{2d}{d+1}+\varepsilon}\mright)
	\\
	&=
	O\mleft({{(r^{d+1})}^{\log_{r^{d+1}}
	N'}}^{\frac{2d}{d+1}+\varepsilon}\mright)
	=
	O\mleft({N'}^{\frac{2d}{d+1}+\varepsilon}\mright)
\end{align*}
\subparagraph{Finally} We can now bound $T_k(N)$ using the upper bound for
$T(N')$,
\begin{displaymath}
	T_{k}(N)
	\le
	\max_{b^{k_i} N_i \le b^{k} N} T_{k_i}(N_i)
	=
	T(b^{k} N)
	=
	O((b^{k} N)^{\frac{2d}{d+1}+\varepsilon})
	=
	2^{O(k)}{N}^{\frac{2d}{d+1}+\varepsilon}.
\end{displaymath}

%% file: sec/03-algorithms/05-nonuniform-implicit.tex
\section{Nonuniform algorithm for 3POL}
\label{sec:algo:implicit:nonuniform}

In this section, we extend the nonuniform algorithm given for explicit 3POL in
\S\ref{sec:algo:explicit:nonuniform} to work for the more general 3POL
problem. We prove the following
\begin{theorem}\label{thm:act-3POL}
    There is a bounded-degree ADT  of depth
    $O(n^{\frac{12}{7}+\varepsilon})$ for 3POL\@.
\end{theorem}

\subparagraph{Idea}
The idea is the same as for explicit 3POL\@. Partition the plane
into $A^*_i \times B^*_j$ cells.
Note that for a fixed $c \in C$, the curve $F(x,y,c)$ intersects $O(\frac ng)$
cells $A^*_i \times B^*_j$. The algorithm is the following: (1) for each cell
$A^*_i \times B^*_j$, sort the real roots of the $F(a,b,z)
\in \mathbb{R}[z]$ taking the union over all $(a,b) \in A_i \times B_j$,
(2) for each $c \in C$,
for each cell $A^*_i \times B^*_j$ intersected by $F(x,y,c)$, binary search on
the sorted order computed in step (1) to find $c$. Step (2) costs $O(n^2
g^{-1}\log g)$. It only remains to implement step (1) efficiently.

\subparagraph{$A \times B$ partition}
We use the same partitioning scheme as before. Hence, counterparts of
Lemma~\ref{lem:intersections}
and
Lemma~\ref{lem:preprocessing}
hold
\begin{lemma}\label{lem:intersections2}
    For a fixed \(c \in C\), the curve $F(x,y,c)=0$ intersects $O(\frac ng)$ cells.
    Moreover, those cells can be computed in $O(\frac ng)$ time.
\end{lemma}
\begin{lemma}\label{lem:preprocessing2}
    If the sets $A,B,C$ can be preprocessed in $S_g(n)$ time so that,
    for any given cell $A^*_i \times B^*_j$ and any given $c \in C$, testing whether
    $c \in \{\,z \st \exists\, (a,b) \in A_i \times B_j
    \text{ such that } F(a,b,z) = 0\,\}$ can be done in
    $O(\log g)$ time, then, 3POL can be solved in
    $S_g(n)+O(\frac{n^2}{g}\log g)$ time.
\end{lemma}

\subparagraph{Interleavings}
Let $\mathcal{P} = (P_1,P_2,\ldots,P_m)$ be a tuple of $m$ univariate polynomials.
Let $\{\,p_{i,1} < p_{i,2} < \cdots < p_{i,\Delta_i}\,\}$
be the set of real roots of $p_i$.
Let $I = ((i_1,j_1),(i_2,j_2),\ldots,(i_{\Delta},j_\Delta))$
be a tuple of pairs of positive integers.
We say that $\mathcal{P}$ realizes $I$ if and only if
$I$ is a permutation of
$\{\,(i,j) \st i \in [m], j \in [\Delta_i]\,\}$, and
for all $t \in [\Delta-1]$, $p_{i_t,j_t} \le p_{i_{t+1},j_{t+1}}$.
When used in this context, we call $I$ an \emph{interleaving}.
Note that (1) the first condition implies $\Delta=\sum_{i=1}^{m}\Delta_i$,
(2) a tuple of polynomials realizes at least one interleaving,
(3) a tuple of polynomials realizes more than one interleaving if some of the
polynomials have common real roots.
We denote by $\mathcal{I}(\mathcal{P})$ the set of interleavings realized by
$\mathcal{P}$.

\subparagraph{$A \times A$ $(b,b')$-partitions}
For a fixed pair $(b,b') \in B\times B$, we partition $\mathbb{R}^2$ into
$(b,b')$-cells so that each cell $\mathcal{C}$ is mapped to a unique
interleaving $I$, and if we take any two points $(a_1,a_1')$ and
$(a_2,a_2')$ inside $\mathcal{C}$, both
$(F(a_1,b,z),F(a_1',b',z))$ and $(F(a_2,b,z),F(a_2',b',z))$ realize
$I$.
Identifying the interleaving associated with each cell of each
$(b,b')$-partition, then locating each $(a,a')$ inside each $(b,b')$-partition
provides the answers to all questions of the form ``Is the $k$th real root of
$F(a,b,z)$ greater than the $\ell$th real root of $F(a',b',z)$, for some $(a,b),
(a',b')\in A_i\times B_j$?''. Those answers are all we need to binary
search for $c$ in the union of the real roots of the $F(a,b,z) \in
\mathbb{R}[z]$ in time $O(\log g)$. Note again that in the nonuniform setting,
we do not sort the roots explicitly, but we must be able to recover the order
from the previous computation steps.

\subparagraph{$\gamma_{b,b'}$ and $\delta_{b}$ curves}
We consider the set of interleavings $\mathcal{I}$ realized by
$(F(x,b,z),F(y,b',z))$ where $z$ is a variable, and $x$ and $y$ are
parameters.
We identify four types of event that can happen when the parameters $x$ and $y$ vary
continuously:
(1) two distinct real roots become common,
(2) a common real root splits into two distinct ones,
(3) a real root appears in one of the polynomials,
and (4) a real root disappears in one of the polynomials.
Note that many of those events can happen concurrently.
By definition of an interleaving, those events are the only ones that can cause
$\mathcal{I}$ to change.

To handle events of the types (1) and (2), we redefine the curves $\gamma_{b,b'}$ introduced
in \S\ref{sec:algo:explicit:nonuniform}%
\footnote{Note that Raz, Sharir, and de Zeeuw~\cite{RSZ15} use the same points
and curves.}
\begin{displaymath}
	\gamma_{b,b'} = \{\, (x,y) \st \exists z\ \text{such that}\
	F(x,b,z)=F(y,b',z)=0 \,\},
\end{displaymath}
that is, $(a,a') \in \gamma_{b,b'}$ if and only if $F(a,b,z)$ and $F(a',b',z)$
have at least one common root. Note that this curve is the curve defined
by the equation
$\res(F(x,b,z),F(y,b,z);z) = 0$,
that is, the set of pairs $(x,y)$ for which the resultant (in $z$) of
$F(x,b,z)$ and $F(y,b,z)$ vanishes. This resultant is a polynomial $\in
\mathbb{R}[x,y]$ of degree at most the square (up to a constant factor) of the
degree of $F$ and can be computed in constant time~\cite{CLO07}.
The following lemma follows by continuity of the manipulated curve
\begin{lemma}\label{lem:cont}
    Let $(a_1,a'_1)$ and $(a_2,a'_2)$ be two points in the plane such
    that there does not exist an interleaving realized by both
    $(F(a_1,b,z),F(a'_1,b',z))$ and
    $(F(a_2,b,z),F(a'_2,b',z))$. Moreover, suppose that those two
    points belong to a connected surface in the plane such that for any point
    $(a,a')$ in that surface, the number of real roots of $F(a,b,z)$ and
    $F(a',b',z)$ is fixed.
    Then the interior of any continuous path from $(a_1,a'_1)$ to $(a_2,a'_2)$ lying in this
    connected surface must intersect $\gamma_{b,b'}$.
\end{lemma}
\begin{proof}
    Let $I_1$ be an interleaving realized by
    $(F(a_1,b,z),F(a'_1,b',z))$ and let $I_2$ be an interleaving realized
    by $(F(a_2,b,z),F(a'_2,b',z))$.
    Because the number of real roots of the polynomials $F(x,b,z)$ and
    $F(y,b',z)$ is fixed for any point $(x,y)$ lying in the connected surface,
    $I_1$ and $I_2$ differ by a nonzero number of swaps.
    Moreover, by contradiction,
    there is a swap that is common to every choice of $I_1$
    and $I_2$.
    Since there is a common swap,
    for some $i,j \in [\deg(F)]$ and without loss of generality,
    the $i$th root of
    $F(a_1,b,z)$ is smaller than the $j$th root of $F(a'_1,b',z)$ whereas the
    $i$th root of $F(a_2,b,z)$ is larger than the $j$th root of
    $F(a'_2,b',z)$. By continuity, on any continuous path from
    $(a_1,a'_1)$ and $(a_2,a'_2)$ there is a point $(a,a')$ such that the
    $i$th root of $F(a,b,z)$ is equal to the $j$th root of $F(a',b',z)$.
    This point cannot be an endpoint of the path, hence,
    the interior of the path intersects $\gamma_{b,b'}$.
\end{proof}
The contrapositive states that, if there exists a continuous path from
$(a_1,a'_1)$ to $(a_2,a'_2)$ whose interior does not intersect the curve
$\gamma_{b,b'}$, then there exists an interleaving realized by both
$(F(a_1,b,z),F(a'_1,b',z))$ and $(F(a_2,b,z),F(a'_2,b',z))$.

To handle events of the types (3) and (4), we define the curve
\begin{displaymath}
	\delta_b = \{\,(x,z) \st F(x,b,z) = 0\,\},
\end{displaymath}
which lies in the $xz$-plane.
\begin{lemma}
    We can partition the $x$ axis of the $xz$-plane into a constant number of
    intervals so that for each interval the number of real roots of
    $F(a,b,z)$ is fixed for all $a$ in this interval.
\end{lemma}
\begin{proof}
    We partition the $xz$-plane into a constant number of vertical
    slabs and lines.
    The $x$ coordinates of vertical tangency points and
    singular points of $\delta_b$ are the values $a$ for which
    a real root of $F(a,b,z)=0$ appears or disappears.
    The number of singular and vertical tangency points of $\delta_b$ is quadratic in $\deg(F)$.
    %
    For each of those points, draw a vertical line that contains the point.
    Those vertical lines partition the $xz$-plane into slabs and lines.
    The number of
    vertical lines we draw is constant because the degree of $F$ is
    constant.
    Figure~\ref{fig:deltab} depicts this drawing.
    The projection of the vertical lines on the $x$ axis produce the desired
    partition (with roughly half of the intervals being singletons).
    Let us name those lines $\delta_b$-lines for further reference.
\end{proof}
We can do a symmetric construction for $F(y,b',z)$ in the $zy$-plane and
get horizontal $\delta_{b'}$-lines.
\begin{lemma}
    We can partition the $y$ axis of the $zy$-plane into a constant number of
    intervals so that for each interval the number of real roots of
    $F(a',b',z)$ is fixed for all $a'$ in this interval.
\end{lemma}

\begin{figure}
	\centering
	\begin{subfigure}[t]{0.5\textwidth}
	\includegraphics[width=\textwidth,trim={-0.11in 0 -0.11in 0},clip]{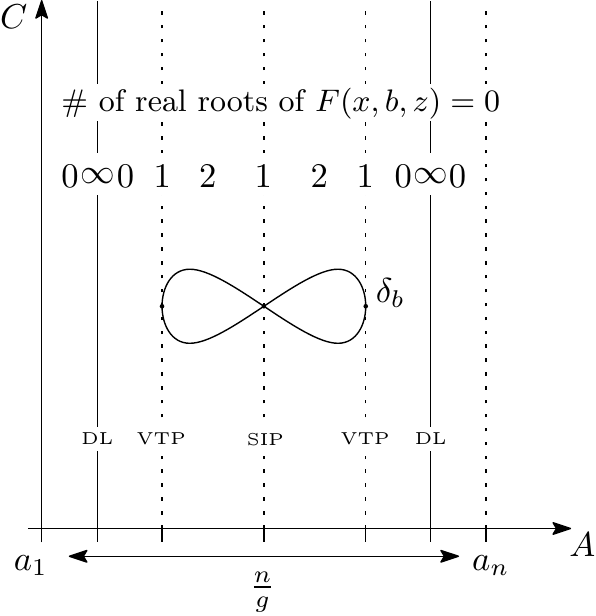}
	\caption{The vertical tangency points (VTP), self-intersection points
	(SIP) and degenerates lines (DL) of $\delta_b$ partition the $A$
	axis into intervals. For all $x$ of the same interval, the
	polynomial $F(x,b,z) \in \mathbb{R}[z]$ has a fixed number of real
	roots.}\label{fig:deltab}
	\end{subfigure}%
	\begin{subfigure}[t]{0.5\textwidth}
	\includegraphics[width=\textwidth]{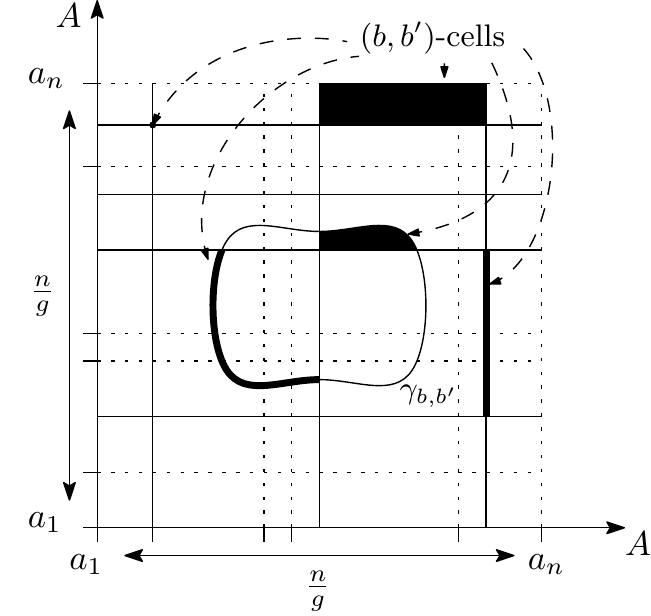}
	\caption{Cells obtained after partitioning the plane using
	the curve $\gamma_{b,b'}$ and the
	$\delta_b$ and $\delta_b'$-lines. The arrows highlight examples of
	$(b,b')$-cells}\label{fig:cells}
	\end{subfigure}
	\caption{Constructions using the $\gamma_{b,b'}$ and $\delta_b$ curves.}\label{fig:delta}
\end{figure}

\subparagraph{Cells of the $(b,b')$-partition}
For a given $(b,b') \in B^2$, let $\Gamma_{b,b'}$ be the set containing
the curve $\gamma_{b,b'}$, the vertical
$\delta_b$-lines and the horizontal $\delta_{b'}$-lines.
The arrangement $\mathcal{A}(\Gamma_{b,b'})$ of those
constant-degree polynomial curves partitions $\mathbb{R}^2$ into a
constant-size set $\mathcal{C}(\Gamma_{b,b'})$ of \emph{$(b,b')$-cells}.
Let $P = \cup_{\gamma \in \Gamma_{b,b'}} \gamma$ and
$\mathcal{C} = \emptyset$. Add all
vertices of $\mathcal{A}(\Gamma_{b,b'})$ to
$\mathcal{C}$. Add each connected component of $P \setminus \mathcal{C}$ to
$\mathcal{C}$. Add each connected component of $\mathbb{R}^2 \setminus P$ to
$\mathcal{C}$. Finally $\mathcal{C}(\Gamma_{b,b'}) := \mathcal{C}$.
Those cells
can be connected surfaces,
pieces of the curve $\gamma_{b,b'}$, pieces of the $\delta_b$- and
$\delta_{b'}$-lines (vertical and horizontal line segments), and
intersections and self-intersection of those curves (vertices).
By construction, all $(b,b')$-cells have the following \emph{invariant property}
\begin{definition}
    A $(b,b')$-cell has the invariant property if, for all points $(a,a')$ in
    that cell, (1) the number of real roots of $F(a,b,z)$ is fixed, (2)
    the number of real roots of $F(a',b',z)$ is fixed,
    and (3) the sorted order of the real roots of $F(a,b,z)$ and $F(a',b',z)$
    is fixed, that is,
    $\mathcal{I}((F(a,b,z),F(a',b',z)))$ is fixed.
\end{definition}
\begin{lemma}
    All $(b,b')$-cells have the invariant property.
\end{lemma}
\begin{proof}
    First, observe that a $(b,b')$-cell that is not a piece of $\gamma_{b,b'}$
    has the invariant property. If that $(b,b')$-cell is a vertex, it contains
    a single point and has thus the invariant property.
    Note that, by construction of $\mathcal{C}(\Gamma_{b,b'})$, a $(b,b')$-cell
    cannot be intersected by a curve of $\Gamma_{b,b'}$ that does not contain
    it. If that $(b,b')$-cell is a piece of a vertical $\delta_b$-line,
    (1) holds because the
    $\delta_b$-line fixes $a$ to some constant, and (2) holds
    because this line segment cannot intersect any
    of the $\delta_{b'}$ lines by construction.
    Assuming this cell is not contained in $\gamma_{b,b'}$, (3)
    holds, and so this cell has the invariant property. A symmetric
    argument settles the case for pieces of horizontal $\delta_{b'}$-lines not
    contained in $\gamma_{b,b'}$.
    Similarly, if that
    $(b,b')$-cell is a connected surface, then it has the invariant property
    because it does not intersect any of the curves in $\Gamma_{b,b'}$.

    Finally, if a $(b,b')$-cell is a piece of $\gamma_{b,b'}$, we make two
    observations. First, this cell cannot intersect any curve of $\Gamma_{b,b'}$ that
    it does not lie in. Hence, similarly to the pieces of a vertical
    $\delta_b$-line, the pieces of a horizontal $\delta_{b'}$-line, and the
    connected surfaces, (1) and (2) hold.
    Second, this cell has two distinct neighbouring connected
    surfaces lying on each of it sides. We just showed that those two
    neighbouring cells have the
    invariant property. Hence the union of our piece of $\gamma_{b,b'}$ with
    those two neighbouring cells is a connected surface as in
    Lemma~\ref{lem:cont}.
    Hence, the ordering of any two real roots
    cannot swap along the piece of
    $\gamma_{b,b'}$. Suppose it would, then this would
    contradict Lemma~\ref{lem:cont}. Hence, (3) holds for pieces of
    $\gamma_{b,b'}$. Hence, those cells have the invariant property.
\end{proof}
This means that once we have computed in which $(b,b')$-cell each
$(a,a')$ point lies, we only need to probe a single point per $(b,b')$-cell
to solve the problem. This partitioning scheme is depicted in
figure~\ref{fig:cells}.

\subparagraph{Preprocessing}
Locate all points $(a,a') \in A_i \times A_i$ for all $A_i$ with respect to all
$\gamma_{b,b'}$ curves, all vertical lines derived from $\delta_b$ and all
horizontal lines derived from $\delta_b'$ for all $(b,b') \in B_j \times B_j$
for all $B_j$ in a single batch using the algorithm described in
Appendix~\ref{sec:algo:point-curves-location}
and the following generalization of Lemma~\ref{lem:dual}:
\begin{lemma}\label{lem:dual:implicit}
Define
\begin{align*}
    \hat{\gamma}_{a,a'} &= \{\, (x,y) \st \res(F(a,x,z),F(a',y,z);z) = 0 \,\},\\
    \hat{\delta}^x_{a} &= \{\, (x,y) \st \res(F(a,x,z),F'_x(a,x,z);z) = 0 \,\},\\
    \hat{\delta}^y_{a'} &= \{\, (x,y) \st \res(F(a',y,z),F'_y(a',y,z);z) = 0 \,\},\\
    \hat{\Gamma}_{a,a'} &= \hat{\gamma}_{a,a'} \cup \hat{\delta}^x_a \cup \hat{\beta}^y_{a'},
\end{align*}
then, locating $(a,a')$ with respect to $\Gamma_{b,b'}$ amounts to
locating $(b,b')$ with respect to $\hat{\Gamma}_{a,a'}$.
\end{lemma}
This takes time
$O({(ng)}^{4/3+\varepsilon})$. For each $(b,b')$, we now know in which
$(b,b')$-cell
each $(a,a')$ point lies and hence the sorted permutation associated with each
cell of the $A \times B$ partition.
We now have the information needed for the binary search in step (2).
The complexity is asymptotically the same as in the explicit case.
This proves Theorem~\ref{thm:act-3POL}.

%% file: sec/03-algorithms/06-uniform-implicit.tex
\section{Uniform algorithm for 3POL}
\label{sec:algo:implicit:uniform}

In this section, we combine the uniform algorithm for explicit 3POL given in
\S\ref{sec:algo:explicit:uniform} with the nonuniform algorithm for
3POL given in Appendix~\ref{sec:algo:implicit:nonuniform} to obtain a uniform
subquadratic algorithm for 3POL\@.
We prove the following
\begin{theorem}\label{thm:implicit:uniform}
	3POL can be solved in
	$O(n^2 {(\log \log n)}^\frac{3}{2} / {(\log n)}^\frac{1}{2})$ time.
\end{theorem}

\subparagraph{Idea}
In the uniform algorithm for explicit 3POL of
\S\ref{sec:algo:explicit:uniform}, we partition the set $A \times B$
into very small sets $A_i \times B_j$, sort the sets $f(A_i \times B_j)$ using
the dominance reporting algorithm of \S\ref{sec:algo:dominance} then
binary search on those sorted sets in order to find a matching $c$.
Here we reuse a similar scheme with the only difference that the sets to sort
are the unions of the real roots of the univariate polynomials
$F(a,b,z)\in\mathbb{R}[z]$ over all $(a,b) \in A_i \times B_j$.
The main difficulty resides in implementing the
equivalent of the certificates of \S\ref{sec:algo:explicit:uniform} to reuse
the dominance reporting algorithm of \S\ref{sec:algo:dominance}. We show how to
implement those certificates using the $\gamma_{b,b'}$ and $\delta_b$ curves
defined in Appendix~\ref{sec:algo:implicit:nonuniform}.

\subparagraph{$A \times B$ partition}
We use the same partitioning scheme as all previous
algorithms, hence Lemma~\ref{lem:intersections2} and
Lemma~\ref{lem:preprocessing2} hold. We apply the same certificate verification
scheme as in \S\ref{sec:algo:explicit:uniform}, hence, the dominance
reporting algorithm of \S\ref{sec:algo:dominance} and the analysis
in \S\ref{sec:algo:explicit:uniform} still apply.

\subparagraph{Preprocessing}
The preprocessing algorithm is essentially the same as
Algorithm~\ref{algo:sfaixbj} with more complex certificates. We explain how to
construct those new certificates. The first part of the explanation consists in
generalizing the definition of a certificate. The rest of the
explanation focuses on the implementation of the verification of those
certificates via Polynomial Dominance Reporting.

\subparagraph{The certificates}
For a fixed pair $(a,b)$, $F(a,b,z) \in \reals[z]$ is a polynomial in $z$ of
degree at most $\deg(F)$. Hence, $F(a,b,z)$ has at most $\deg(F)$ real roots.
For each cell $A_i^* \times B_j^*$, let
\begin{displaymath}
	A_i\times B_j=\{\,
		(a_{i,1},b_{j,1}),
		(a_{i,1},b_{j,2}),
		\ldots,
		(a_{i,2},b_{j,1}),
		(a_{i,2},b_{j,2}),
		\ldots,
		(a_{i,g},b_{j,g})
	\,\}.
\end{displaymath}
Let $\rho
\colon {[g]}^2 \to \{\,0,1,\ldots,\deg(F)\,\}$ be a function that maps a pair $(k,l)$ to the
number of real roots of $F(a_{i,k},b_{j,l},z)$. Let $\Sigma_\rho = \sum_{(i,j)
\in {[g]}^2} \rho(i,j) \le \deg(F) g^2$.
Given a function $\rho$, let
$\pi\colon\,[\Sigma_\rho]\to {[g]}^2 \times
\{\,0,1,\ldots,\deg(F)\,\}$ be a permutation of the union of the real roots of all
$g^2$ polynomials $
F(a_{i,1},b_{j,1},z),
F(a_{i,1},b_{j,2},z),
\ldots,
F(a_{i,2},b_{j,1},z),
F(a_{i,2},b_{j,2},z),
\ldots,
F(a_{i,g},b_{j,g},z)$, where the number of real roots of each polynomial is
prescribed by $\rho$.
Decompose $\pi = (\pi_r,\pi_c,\pi_s)$ into
row,
column and real root number functions $\pi_r,\pi_c\colon\,[\Sigma_\rho]\to[g]$,
and $\pi_s\colon\,[\Sigma_\rho]\to\{\,0,1,\ldots,\deg(F)\,\}$.
Let $\diamond(a,b,s)$ denote the $s$th real root of $F(a,b,z)$.
To fix the permutation of the union of the real roots of all
$g^2$ polynomials, we define the following
interleaving certificate with $\Sigma_\rho - 1$ inequalities, for each possible
function $\rho$ and permutation $\pi$
\begin{displaymath}
	\Phi_{\rho,\pi} :=
	\diamond(a_{i,\pi_r(1)},b_{j,\pi_c(1)},\pi_s(1))
	\le
	\cdots
	\le
	\diamond(a_{i,\pi_r(\Sigma_\rho)},b_{j,\pi_c(\Sigma_\rho)},\pi_s(\Sigma_\rho)).
\end{displaymath}
To fix the number of real roots each of the $g^2$ polynomials can have, we
define the following cardinality certificate for each function $\rho$
\begin{displaymath}
	\Psi_{\rho} :=
	\bigwedge_{(k,l)\in{[g]}^2}\,\,
	F(a_{i,k},b_{j,l},z)\,\,\text{has $\rho(k,l)$ real roots}.
\end{displaymath}
For each possible function $\rho$ and permutation $\pi$ we define the
certificate $\Upsilon_{\rho,\pi} := \Psi_\rho \land \Phi_{\rho,\pi}$ that
fixes both the number of real roots each polynomial has and the permutation of
those real roots.
The total number of certificates $\Upsilon_{\rho,\pi}$ is
$\sum_{\rho\colon {[g]}^2 \to \{\,0,1,\ldots,\deg(F)\,\}} {\Sigma_\rho!}$
which is of the order of ${(g^2)}^{O(g^2)}$.

Finally, we need to handle the edge cases where a polynomial $F(a,b,z)$ is the
zero polynomial. In that case, $F(a,b,z)$ cancels for all $z \in
\mathbb{R}$. Hence, all planar curves $F(x,y,c)=0$ go through $(a,b)$ and we
can immediately accept the 3POL instance. To capture
those edge cases, we will check the following certificate before running the main
algorithm
\begin{displaymath}
	\Omega := \bigvee_{(k,l)\in{{[g]}^2}} F(a_{i,k},b_{j,l},z)\,\,\text{is the zero polynomial}.
\end{displaymath}

We can check if $\Omega$ holds for any cell $A_i \times B_j$ in $O(n \log n)$ time.
For each $b \in B$ binary search for a $a \in A$ that lies on a vertical line
component of $\delta_b$.

If this certificate is verified we accept and halt. Otherwise we can safely run
the main algorithm.

\subparagraph{$A \times A$ $(b,b')$-partitions}
For each $B_j$ and for each $(b,b') \in B_j^2$ compute a partition of the
$A \times A$ grid according to the $(b,b')$-cells defined by $\Gamma_{b,b'}$
--- see Appendix~\ref{sec:algo:implicit:nonuniform}. For each $(b,b')$-cell of
that partition, pick a sample point $(a,a')$, compute the interleaving
$\mathcal{I}((F(a,b,z),F(a',b',z)))$. Store that information for future lookup.
All this takes $O(ng)$ time.

\subparagraph{PDR instance for $\Psi_\rho$}
For a fixed pair $(a,b)$, suppose $F(a,b,z)$ has $r$ real roots. Then $a$ must
lie in one of the open intervals or be one of the breaking points defined by
the VTP, SIP and DL of $\delta_b$ that fixes the number of real roots of
$F(a,b,z)$ to $r$.
Hence $\Psi_\rho$ can be rewritten as follows
\begin{displaymath}
	\Psi_{\rho} =
	\bigwedge_{(k,l)\in{[g]}^2}\,\,
	\left(\bigvee_{[u,v]\in \mathcal{I}_{\rho(k,l)}} u < a_{i,k} < v\right)
	\bigvee
	\left(\bigvee_{w\in \mathcal{B}_{\rho(k,l)}} a_{i,k} = w\right)
\end{displaymath}
where $\mathcal{I}_{\rho(k,l)}$ denotes the set of intervals fixing the number
of real roots of $F(a_{i,k},b_{j,l},z)$ to $\rho(k,l)$, and
$\mathcal{B}_{\rho(k,l)}$ denotes the set of breaking points fixing the number
of real roots of $F(a_{i,k},b_{j,l},z)$ to $\rho(k,l)$.

The PDR algorithm can only check conjunctions of polynomial inequalities.
However, we can transform $\Psi_{\rho}$ into disjunctive normal form~(DNF) by
splitting the certificate into distinct branches, each consisting of a
conjunction of polynomial inequalities.
Since the number of intervals and breaking points
considered above is constant for each pair $(k,l)$, the number of branches to
test is $2^{O(g^2)}$.

For each $A_i$ we have thus a single vector of reals
\begin{displaymath}
	p_i = (
		a_{i,1} , a_{i,1},
		a_{i,2}, a_{i,2},
		\ldots,
		a_{i,g} , a_{i,g}
	),
\end{displaymath}
and for each $B_j$ we have $2^{O(g^2)}$ vectors of linear inequalities
\begin{displaymath}
	q_j = (
		x \sigma_{u_{1,1}} u_{1,1}, x \sigma_{v_{1,1}} v_{1,1},
		x \sigma_{u_{1,2}} u_{1,2}, x \sigma_{v_{1,2}} v_{1,2},
		\ldots,
		x \sigma_{u_{g,g}} u_{g,g}, x \sigma_{v_{g,g}} v_{g,g},
	),
\end{displaymath}
where each $(\sigma_{u_{k,l}}, u_{k,l}, \sigma_{v_{k,l}}, v_{k,l})$ is an element of
\begin{displaymath}
	\{\, ( > , u , < , v ) \st (u,v) \in \mathcal{I}_{\rho(k,l)}\,\}
	\cup
	\{\, ( = , w , = , w ) \st w \in \mathcal{B}_{\rho(k,l)}\,\}.
\end{displaymath}

For a fixed function $\rho$,
the sets of vectors $p_i$ and $q_j$ is a valid PDR instance
of size $N = ng^{-1} 2^{O(g)}$ and with parameter $k = 2g^2$
that will
output all cells $A^*_i \times B^*_j$ such that $F(a_{i,k},b_{j,l},z) \in
\mathbb{R}[z]$ has exactly $\rho(k,l)$ real roots for all $(a_{i,k},a_{j,l}) \in A_i
\times B_j$.

\subparagraph{PDR instance for $\Phi_{\rho,\pi}$}

For fixed pairs $(a,b)$ and $(a',b')$, suppose the $s$-th real root of
$F(a,b,z)$ is smaller or equal to the $q$-th real root of $F(a,b,z)$. Then,
$(a,a')$ must lie in a $(b,b')$-cell that orders the $s$-th root of $F(x,b,z)$
before the $q$-th root of $F(y,b',z)$ for all points $(x,y)$ in that cell.

Hence $\Phi_{\rho,\pi}$ can be rewritten as follows
\begin{displaymath}
	\Phi_{\rho,\pi} =
	\bigwedge_{t\in[\Sigma_\rho-1]}\,\,
	\bigvee_{C \in \mathcal{C}_{\rho,\pi,t}}
	( a_{i,\pi_r(t)}, a_{i,\pi_r(t+1)} ) \in C
\end{displaymath}
where $\mathcal{C}_{\rho,\pi,t}$ denotes the set of $(b,b')$-cells
fixing the number of real roots of $F(a_{i,\pi_r(t)},b_{j,\pi_c(t)},z)$
to $\rho(\pi_r(t),\pi_c(t))$,
fixing the number of real roots of $F(a_{i,\pi_r(t+1)},b_{j,\pi_c(t+1)},z)$
to $\rho(\pi_r(t+1),\pi_c(t+1))$,
and ordering the $\pi_s(t)$-th root of $F(a_{i,\pi_r(t)},b_{j,\pi_c(t)},z)$
before the $\pi_s(t+1)$-th root of $F(a_{i,\pi_r(t+1)},b_{j,\pi_c(t+1)},z)$.

The PDR algorithm can only check conjunctions of polynomial inequalities.
However, we can transform $\Phi_{\rho,\pi}$ in DNF as we did for $\Psi_{\rho}$.
Again the number of cells considered above is constant for each $t$, the
description of each cell is constant, hence, the number of branches to test is
$2^{O(g^2)}$.

For each $A_i$ we have thus a single vector of $2$-dimensional points
\begin{displaymath}
	p_i = (
		\underbrace{%
			(a_{i,\pi_r(1)},a_{i,\pi_r(2)}),
			\ldots,
			(a_{i,\pi_r(1)},a_{i,\pi_r(2)})
		}_{w},
		\ldots,
		\underbrace{%
			(a_{i,\pi_r(\Sigma_\rho-1)},a_{i,\pi_r(\Sigma_\rho)}),
			\ldots,
			(a_{i,\pi_r(\Sigma_\rho-1)},a_{i,\pi_r(\Sigma_\rho)})
		}_{w}
	),
\end{displaymath}
where $w$ is the size of the largest description of a $(b,b')$-cell $C$,
and for each $B_j$ we have $2^{O(g^2)}$ vectors of polynomial inequalities,
\begin{displaymath}
	q_j = (
		h_{1,1}(x,y) \sigma_{1,1} 0,
		\ldots,
		h_{1,w}(x,y) \sigma_{1,w} 0,
		\ldots,
		h_{\Sigma_\rho-1,1}(x,y) \sigma_{\Sigma_\rho-1,1} 0,
		\ldots,
		h_{\Sigma_\rho-1,w}(x,y) \sigma_{\Sigma_\rho-1,w} 0
	),
\end{displaymath}
where each $(h_{t,1}(x,y) \sigma_{t,1} 0, \ldots,
h_{t,w}(x,y) \sigma_{t,w} 0)$ is an element of
$\{\, \text{desc}(C) \st C \in \mathcal{C}_{\rho,\pi,t}\,\}$,
where $\text{desc}(C)$ is the description of the cell $C$ given as a
certificate of belonging to $C$ in the form of a Tarski sentence.
The description of each $(b,b')$-cell is padded with
its last component so that it has length $w$.

For a fixed function $\rho$,
for a fixed function $\pi$,
the sets of vectors $p_i$ and $q_j$ is a valid PDR instance
of size $N = ng^{-1} 2^{O(g)}$ and with parameter $k = \Theta(g^2)$
that will
output all cells $A^*_i \times B^*_j$ such that
the number of real roots of $F(a_{i,\pi_r(t)},b_{j,\pi_c(t)},z)$
is $\rho(\pi_r(t),\pi_c(t))$,
the number of real roots of $F(a_{i,\pi_r(t+1)},b_{j,\pi_c(t+1)},z)$
is $\rho(\pi_r(t+1),\pi_c(t+1))$,
and the $\pi_s(t)$-th root of $F(a_{i,\pi_r(t)},b_{j,\pi_c(t)},z)$
comes before the $\pi_s(t+1)$-th root of
$F(a_{i,\pi_r(t+1)},b_{j,\pi_c(t+1)},z)$,
for all $t \in [\Sigma_\rho-1]$.

\subparagraph{PDR instance for $\Upsilon_{\rho,\pi}$}

We can combine the certificates given above for $\Psi_\rho$ and
$\Phi_{\rho,\pi}$ to obtain the ones for $\Upsilon_{\rho,\pi}$: concatenate the
$p_i$ and $q_j$ together (add a dummy $y$ variable for the $p_i$ and $q_j$ of
$\Psi_\rho$).
For a fixed function $\rho$,
for a fixed function $\pi$,
the sets of vectors $p_i$ and $q_j$ is a valid PDR instance
of size $N = ng^{-1} 2^{O(g)}$ and with parameter $k = \Theta(g^2)$
that will
output all cells $A^*_i \times B^*_j$ such that
$F(a_{i,k},b_{j,l},z) \in
\mathbb{R}[z]$ has exactly $\rho(k,l)$ real roots for all $(a_{i,k},a_{j,l}) \in A_i
\times B_j$,
and the $\pi_s(t)$-th root of $F(a_{i,\pi_r(t)},b_{j,\pi_c(t)},z)$
comes before the $\pi_s(t+1)$-th root of $F(a_{i,\pi_r(t+1)},b_{j,\pi_c(t+1)},z)$
for all $t \in [\Sigma_\rho-1]$.
The rest of the analysis in \S\ref{sec:algo:explicit:uniform} applies.
This proves Theorem~\ref{thm:implicit:uniform}.

%% file: sec/04-applications/00-intro.tex
\section{Applications}\label{sec:applications}

To illustrate the expressive power of 3POL, we give a few applications.

%% file: sec/04-applications/01-gpt.tex
\subsection{General position testing for points on curves}


The following is a corollary of Theorem~\ref{thm:rsz15:col} in Raz, Sharir and de Zeeuw~\cite{RSZ15}
\begin{corollary}[Raz, Sharir and de Zeeuw~\cite{RSZ15}]
	Any $n$ points on an irreducible algebraic curve of degree $d$ in
	$\mathbb{C}^2$ determine
	$\tilde{O_d}(n^{\frac{11}{6}})$ proper collinear triples, unless the curve is a line or a cubic.
\end{corollary}

An interesting application of our results is the existence of subquadratic
nonuniform and uniform algorithms for the computational version of
this corollary.
\begin{problem}[GPT on curves]
	Let $C_1, C_2, C_3$ be three (not necessarily distinct) parameterized
	constant-degree polynomial curves in $\mathbb{R}^2$, so that each
	$C_i$ can be written $(g_i(t),h_i(t))$ for some polynomials of
	constant degree $g_i,h_i$.
	Given three $n$-sets $S_1 \subset C_1, S_2 \subset C_2, S_3
	\subset C_3$, decide whether there exist any collinear
	triple of points in $S_1 \times S_2 \times S_3$.
\end{problem}
\begin{theorem}\label{thm:gpt-to-3pol}
	 GPT on curves reduces linearily to 3POL\@.
\end{theorem}
\begin{proof}
For each set $S_i$, construct the set $T_i = \{\,t
\st p \in S_i, p = (g_{i}(t),h_{i}(t))\,\}$.
Testing whether there exists a collinear triple
$((g_1(t_1),h_1(t_1)),(g_2(t_2),h_2(t_2)),(g_3(t_3),h_3(t_3)))
\in S_1 \times S_2 \times S_3$ amounts to testing whether any determinant
\begin{displaymath}
\begin{vmatrix}
g_1(t_1)&h_1(t_1)&1\\
g_2(t_2)&h_2(t_2)&1\\
g_3(t_3)&h_3(t_3)&1
\end{vmatrix}
\end{displaymath}
equals zero.
This determinant is a trivariate constant-degree polynomial in
$\mathbb{R}[t_1,t_2,t_3]$.
Solving the original problem amounts thus to deciding whether this polynomial
cancels for any triple $(t_1,t_2,t_3) \in T_1 \times T_2 \times T_3$.
\end{proof}

Note that a similar polynomial predicate exists for testing collinearity in
higher dimension.
\begin{lemma}
Let $p = (p_1,p_2,\ldots,p_d)$, $q = (q_1,q_2,\ldots,q_d)$,
and $r=(r_1,r_2,\ldots,r_d)$ be three points in $\mathbb{R}^d$, then $p$, $q$, and
$r$ are collinear if and only if
\begin{displaymath}
{\left[\sum_{i=1}^{d}(p_i-r_i)(q_i-p_i)\right]}^2 -
\left[\sum_{i=1}^{d}{(p_i-r_i)}^2\right]\left[\sum_{i=1}^{d}{(q_i-p_i)}^2\right] =
0.
\end{displaymath}
\end{lemma}
\begin{proof}
Let $a = (a_1,a_2,\ldots,a_d)$, $b = (b_1,b_2,\ldots,b_d)$, and
$c=(c_1,c_2,\ldots,c_d)$ be three points in $\mathbb{R}^d$. The points $a$, $b$,
and $c$ are collinear if and only if $c = a + \lambda (b-a)$ for some unique
$\lambda \in \mathbb{R}$ that is
\begin{gather*}
	(a-c) + \lambda (b-a) = 0\\
	(a_i-c_i) + \lambda (b_i-a_i) = 0, \forall i \in [d]\\
	\sum_{i=1}^{d}{\left[(a_i-c_i) + \lambda (b_i-a_i)\right]}^2 = 0\\
\sum_{i=1}^{d}\left[{(b_i-a_i)}^2 \lambda^2 + 2(a_i-c_i)(b_i-a_i) \lambda +
{(a_i-c_i)}^2\right] = 0\\
\underbrace{\left[\sum_{i=1}^{d}{(b_i-a_i)}^2\right]}_{A} \lambda^2 +
\underbrace{\left[2\sum_{i=1}^{d}(a_i-c_i)(b_i-a_i)\right]}_{B} \lambda +
\underbrace{\left[\sum_{i=1}^{d}{(a_i-c_i)}^2\right]}_{C} = 0\\
\lambda = \frac{-B \pm \sqrt{B^2-4AC}}{2A}
\end{gather*}

For $\lambda$ to exist and be unique $B^2-4AC$ must be zero. Hence, $a$, $b$, and $c$ are collinear if and only if

\begin{displaymath}
	{\left[2\sum_{i=1}^{d}(a_i-c_i)(b_i-a_i)\right]}^2
	- 4 \left[\sum_{i=1}^{d}{(a_i-c_i)}^2\right]\left[\sum_{i=1}^{d}{(b_i-a_i)}^2\right]
	= 0
\end{displaymath}
\end{proof}

Moreover, the improvement that we obtain in the time complexity of 3POL
can be exploited to boost the number of curves we pick the points from.
\begin{theorem}
	Let $C_1,C_2,\ldots,C_k$ be $k=o({(\log n)}^\frac{1}{6}/{(\log \log n)}^\frac{1}{2})$
	(not necessarily distinct)
	constant-degree polynomial curves in $\mathbb{R}^d$.
	Given $k$ $n$-sets $S_1 \subset C_1, S_2 \subset C_2, \ldots, S_k
	\subset C_k$, deciding whether there exists any collinear
	triple of points in any triple of sets $S_{i_1} \times S_{i_2} \times
	S_{i_3}$ can be solved in subquadratic time.
\end{theorem}
\begin{proof}
	Solve a 3POL problem for each choice of $S_{i_1} \times S_{i_2} \times
	S_{i_3}$. There are $o({(\log n)}^\frac{1}{2}/{(\log \log n)}^\frac{3}{2})$
	such choices.
\end{proof}

%% file: sec/04-applications/02-unit-circles.tex
\subsection{Incidences on unit circles}

Raz, Sharir and Solymosi~\cite{RSS15} mention the following problem as a
special case of the framework they introduce.
Let $p_1,p_2,p_3$ be three distinct points in the plane, and, for $i=1,2,3$,
let $\mathcal{C}_i$ be a family of $n$ unit circles (a circle of radius $1$)
that pass through $p_i$.  Their goal is to obtain an upper bound on the number
of \emph{triple points}, which are points that are incident to a circle of each
family.
They prove:
\begin{theorem}
	Let $p_1,p_2,p_3$ be three distinct points in the plane, and, for $i=1,2,3$,
	let $\mathcal{C}_i$ be a family of $n$ unit circles that pass through $p_i$.
	Then the number of points incident to a circle of each family is $O(n^{11/6})$.
\end{theorem}

They observe that the following dual formulation is equivalent to their
original problem:
\begin{theorem}
	Let $C_1,C_2,C_3$ be three unit circles in $\mathbb{R}^2$, and, for each
	$i=1,2,3$, let $S_i$ be a set of $n$ points lying on $C_i$. Then the number of
	unit circles, spanned by triples of points in $S_1 \times S_2 \times S_3$, is
	$O(n^{11/6})$.
\end{theorem}

Our new algorithms indeed allow us to solve the decision version of their
problems in subquadratic time.
\begin{problem}[Unit Circles Spanned by Points on Three Unit Circles (UCSPTUC)]
	Let $C_1,C_2,C_3$ be three unit circles in $\mathbb{R}^2$ with centers
	$c_1,c_2,c_3$, and, for each $i=1,2,3$, let $S_i =
	\{\,(x_{i,1},y_{i,1}),(x_{i,2},y_{i,2}),\ldots,(x_{i,n},y_{i,n})\,\}$ be a
	set of $n$ points lying on $C_i$. Decide whether any triple of point
	$(a,b,c) \in S_1 \times S_2 \times S_3$ spans a unit circle.
\end{problem}

\begin{theorem}
	UCSPTUC can be solved in
	$O(n^2 {(\log \log n)}^\frac{3}{2} / {(\log n)}^\frac{1}{2})$ time.
\end{theorem}

\begin{proof}
	Without loss of generality, assume all input points lie on the right
	$y$-monotone arc of their respective circle. All other seven cases can be
	handled similarly. We can also assume that no input point is the top or
	bottom vertex of its circle, rotating the plane if necessary.

	Given three points $p,q,r$, let
	\begin{displaymath}
		x=\lVert p - q \rVert,
		X=x^2,
		y=\lVert p - r \rVert,
		Y=y^2,
		z=\lVert q - r \rVert,
		Z=z^2
	\end{displaymath}
	Testing if the three points $a,b,c$ span a unit circle amounts to testing
	whether
	\begin{displaymath}
		X^2 + Y^2 + Z^2 - 2XY - 2XZ - 2YZ + XYZ = 0.
	\end{displaymath}

	The fact that the input points lie on the right $y$-monotone arc of unit
	circles of centers $c_1,c_2,c_3$ allows us to get down to a single variable
	per point. Let $c_i = (c_i^x,c_i^y)$ and
	$t_{i,j} = \sqrt{\frac{1 - x_{i,j} + c_i^x}{1 + x_{i,j} - c_i^x}}$. Then
	\begin{displaymath}
		(x_{i,j},y_{i,j}) = c_i + \left(\frac{1-t_{i,j}^2}{1+t_{i,j}^2},\frac{2
		t_{i,j}}{1+t_{i,j}^2}\right).
	\end{displaymath}

	Combining those two observations with some algebraic manipulations, one can
	show that there exists some trivariate polynomial $F$ of degree at most
	$24$ that cancels on $t_1,t_2,t_3$ when the points $c_1
	+\left(\frac{1-t_1^2}{1+t_1^2},\frac{2t_1}{1+t_1^2}\right)$, $c_2
	+\left(\frac{1-t_2^2}{1+t_2^2},\frac{2t_2}{1+t_2^2}\right)$, and $c_3
	+\left(\frac{1-t_3^2}{1+t_3^2},\frac{2t_3}{1+t_3^2}\right)$ span a unit
	circle.

	Hence, the sets
	$\{\,t_{1,1},t_{1,2},\ldots,t_{1,n}\,\}$,
	$\{\,t_{2,1},t_{2,2},\ldots,t_{2,n}\,\}$, and
	$\{\,t_{3,1},t_{3,2},\ldots,t_{3,n}\,\}$
	together with $F$ gives an instance of 3POL we can solve in subquadratic
	time with our new algorithms.

	Unfortunately, the computation $\sqrt{x}$ is not allowed in our model, and
	so, we cannot compute $t_{i,j}$.
	However, we can generalize the 3POL problem to make it fit:
	\begin{problem}[Modified 3POL]
		Let $F \in \mathbb{R}[x,y,z]$ be a trivariate polynomial of constant degree,
		given three sets $A$, $B$, and $C$, each containing $n$ real numbers, decide
		whether there exist $a \in A$, $b \in B$, and $c \in C$ such that
		\begin{displaymath}
			\exists t_1,t_2,t_3
			t_1^2 = a \land
			t_2^2 = b \land
			t_3^2 = c \land
			F(t_1,t_2,t_3) = 0.
		\end{displaymath}
	\end{problem}

	The sets of numbers (all computable in our models)
	$\{\,t_{1,1}^2,t_{1,2}^2,\ldots,t_{1,n}^2\,\}$,
	$\{\,t_{2,1}^2,t_{2,2}^2,\ldots,t_{2,n}^2\,\}$, and
	$\{\,t_{3,1}^2,t_{3,2}^2,\ldots,t_{3,n}^2\,\}$
	together with $F$ give an instance of this modified version of 3POL\@.

	We can tweak our algorithms so that they work for this new version of
	3POL\@.
	For each decision we make on the FOTR, we prefix it with an existential
	quantifier and a condition of the type $t_i^2=x$, with $x$ the square of
	$t_i$, when we want to reference $t_i$ in the formula we want to test.

	This new algorithm answers positively if and only if the original problem
	contains a triple of points spanning a unit circle.
\end{proof}

%% file: sec/04-applications/03-unit-triangles.tex
\subsection{Points spanning unit triangles}

A similar problem, namely counting the number of input point triples spanning
an area $S$ triangle (provided they lie on a few curves), can also easily be
reduced to 3POL\@.
The polynomial to look at in this case is
\begin{displaymath}
	F(x,y,z) = X^2 + Y^2 + Z^2 - 2XY - 2XZ - 2YZ + 16 S^2.
\end{displaymath}

Note that when the input points lie in the plane, the number of solutions is
more than quadratic~\cite{RS15,RSS15b}.

%% file: 3pol.bbl
\begin{thebibliography}{10}

\bibitem{AV14}
Amir Abboud and Virginia~Vassilevska Williams.
\newblock Popular conjectures imply strong lower bounds for dynamic problems.
\newblock In {\em {FOCS}}, pages 434--443. {IEEE} Computer Society, 2014.

\bibitem{AVW14}
Amir Abboud, Virginia~Vassilevska Williams, and Oren Weimann.
\newblock Consequences of faster alignment of sequences.
\newblock In {\em {ICALP} {(1)}}, volume 8572 of {\em LNCS}, pages 39--51,
  2014.

\bibitem{AVY15}
Amir Abboud, Virginia~Vassilevska Williams, and Huacheng Yu.
\newblock Matching triangles and basing hardness on an extremely popular
  conjecture.
\newblock In {\em {STOC}}, pages 41--50. {ACM}, 2015.

\bibitem{AC05}
Nir Ailon and Bernard Chazelle.
\newblock Lower bounds for linear degeneracy testing.
\newblock {\em J. {ACM}}, 52(2):157--171, 2005.

\bibitem{ACLL14}
Amihood Amir, Timothy~M. Chan, Moshe Lewenstein, and Noa Lewenstein.
\newblock On hardness of jumbled indexing.
\newblock In {\em {ICALP} {(1)}}, volume 8572 of {\em LNCS}, pages 114--125,
  2014.

\bibitem{BDP08}
Ilya Baran, Erik~D. Demaine, and Mihai P{\u a}tra{\cb s}cu.
\newblock Subquadratic algorithms for {3SUM}.
\newblock {\em Algorithmica}, 50(4):584--596, 2008.

\bibitem{BH01}
Gill Barequet and Sariel Har{-}Peled.
\newblock Polygon containment and translational min {H}ausdorff distance
  between segment sets are {3SUM}-hard.
\newblock {\em Int. J. Comput. Geometry Appl.}, 11(4):465--474, 2001.

\bibitem{BPR96b}
Saugata Basu, Richard Pollack, and Marie{-}Fran{\c{c}}oise Roy.
\newblock Computing roadmaps of semi-algebraic sets (extended abstract).
\newblock In {\em {STOC}}, pages 168--173. {ACM}, 1996.

\bibitem{BPR06}
Saugata Basu, Richard Pollack, and Marie-Fran{\c{c}}oise Roy.
\newblock {\em Algorithms in real algebraic geometry}, volume~10 of {\em
  Algorithms and Computation in Mathematics}.
\newblock Springer, 2006.

\bibitem{BCDEHILPT14}
David Bremner, Timothy~M. Chan, Erik~D. Demaine, Jeff Erickson, Ferran Hurtado,
  John Iacono, Stefan Langerman, Mihai Patrascu, and Perouz Taslakian.
\newblock Necklaces, convolutions, and {X+Y}.
\newblock {\em Algorithmica}, 69(2):294--314, 2014.

\bibitem{BCM99}
Herv{\'{e}} Br{\"{o}}nnimann, Bernard Chazelle, and Jir{\'{\i}} Matou{\v s}ek.
\newblock Product range spaces, sensitive sampling, and derandomization.
\newblock {\em {SIAM} J. Comput.}, 28(5):1552--1575, 1999.

\bibitem{CGIMPS16}
Marco~L. Carmosino, Jiawei Gao, Russell Impagliazzo, Ivan Mihajlin, Ramamohan
  Paturi, and Stefan Schneider.
\newblock Nondeterministic extensions of the strong exponential time hypothesis
  and consequences for non-reducibility.
\newblock In {\em {ITCS}}, pages 261--270. {ACM}, 2016.

\bibitem{CJ12}
Bob~F Caviness and Jeremy~R Johnson.
\newblock {\em Quantifier elimination and cylindrical algebraic decomposition}.
\newblock Springer, 2012.

\bibitem{Cha08}
Timothy~M. Chan.
\newblock All-pairs shortest paths with real weights in ${O}(n^3/\log n)$ time.
\newblock {\em Algorithmica}, 50(2):236--243, 2008.

\bibitem{CEGS91}
Bernard Chazelle, Herbert Edelsbrunner, Leonidas~J. Guibas, and Micha Sharir.
\newblock A singly exponential stratification scheme for real semi-algebraic
  varieties and its applications.
\newblock {\em Theor. Comput. Sci.}, 84(1):77--105, 1991.

\bibitem{CM96}
Bernard Chazelle and Jir{\'{\i}} Matousek.
\newblock On linear-time deterministic algorithms for optimization problems in
  fixed dimension.
\newblock {\em J. Algorithms}, 21(3):579--597, 1996.

\bibitem{C75}
George~E. Collins.
\newblock Hauptvortrag: Quantifier elimination for real closed fields by
  cylindrical algebraic decomposition.
\newblock In {\em Automata Theory and Formal Languages}, volume~33 of {\em
  LNCS}, pages 134--183. Springer, 1975.

\bibitem{CLO07}
David Cox, John Little, and Donal O'shea.
\newblock {\em Ideals, Varieties, and Algorithms: An Introduction to
  Computational Algebraic Geometry and Commutative Algebra}.
\newblock Undergraduate Texts in Mathematics. Springer, 2007.

\bibitem{DH88}
James~H. Davenport and Joos Heintz.
\newblock Real quantifier elimination is doubly exponential.
\newblock {\em J. Symb. Comput.}, 5(1/2):29--35, 1988.

\bibitem{EGPPSS92}
Herbert Edelsbrunner, Leonidas~J. Guibas, J{\'{a}}nos Pach, Richard Pollack,
  Raimund Seidel, and Micha Sharir.
\newblock Arrangements of curves in the plane - topology, combinatorics and
  algorithms.
\newblock {\em Theor. Comput. Sci.}, 92(2):319--336, 1992.

\bibitem{ER00}
Gy{\"{o}}rgy Elekes and Lajos R{\'{o}}nyai.
\newblock A combinatorial problem on polynomials and rational functions.
\newblock {\em J. Comb. Theory, Ser. {A}}, 89(1):1--20, 2000.

\bibitem{ES12}
Gy{\"{o}}rgy Elekes and Endre Szab{\'{o}}.
\newblock How to find groups? (and how to use them in {Erd{\H{o}}s} geometry?).
\newblock {\em Combinatorica}, 32(5):537--571, 2012.

\bibitem{E99}
Jeff Erickson.
\newblock Lower bounds for linear satisfiability problems.
\newblock {\em Chicago J. Theor. Comput. Sci.}, 1999.

\bibitem{F76}
Michael~L. Fredman.
\newblock How good is the information theory bound in sorting?
\newblock {\em Theor. Comput. Sci.}, 1(4):355--361, 1976.

\bibitem{F15}
Ari Freund.
\newblock Improved subquadratic {3SUM}.
\newblock {\em Algorithmica}, pages 1--19, 2015.

\bibitem{GO95}
Anka Gajentaan and Mark~H. Overmars.
\newblock On a class of ${O}(n^2)$ problems in computational geometry.
\newblock {\em Comput. Geom.}, 5:165--185, 1995.

\bibitem{GS15}
Omer Gold and Micha Sharir.
\newblock Improved bounds for {3SUM}, $k$-{SUM}, and linear degeneracy.
\newblock {\em ArXiv e-prints}, 2015.
\newblock \href{https://arxiv.org/abs/1512.05279}{arXiv:1512.05279 [cs.DS]}.

\bibitem{GP14}
Allan Gr{\o}nlund and Seth Pettie.
\newblock Threesomes, degenerates, and love triangles.
\newblock In {\em {FOCS}}, pages 621--630. {IEEE} Computer Society, 2014.

\bibitem{Har13}
Joe Harris.
\newblock {\em Algebraic geometry: a first course}, volume 133.
\newblock Springer, 2013.

\bibitem{Har77}
Robin Hartshorne.
\newblock {\em Algebraic geometry}, volume~52.
\newblock Springer, 1977.

\bibitem{HKNS15}
Monika Henzinger, Sebastian Krinninger, Danupon Nanongkai, and Thatchaphol
  Saranurak.
\newblock Unifying and strengthening hardness for dynamic problems via the
  online matrix-vector multiplication conjecture.
\newblock In {\em {STOC}}, pages 21--30. {ACM}, 2015.

\bibitem{KPP16}
Tsvi Kopelowitz, Seth Pettie, and Ely Porat.
\newblock Higher lower bounds from the {3SUM} conjecture.
\newblock In {\em {SODA}}, pages 1272--1287. {SIAM}, 2016.

\bibitem{Ma93}
Jir{\'{\i}} Matou{\v s}ek.
\newblock Range searching with efficient hierarchical cutting.
\newblock {\em Discrete {\&} Computational Geometry}, 10:157--182, 1993.

\bibitem{M95}
Jir{\'{\i}} Matousek.
\newblock Approximations and optimal geometric divide-an-conquer.
\newblock {\em J. Comput. Syst. Sci.}, 50(2):203--208, 1995.

\bibitem{M96}
Jir{\'{\i}} Matou{\v s}ek.
\newblock Derandomization in computational geometry.
\newblock {\em J. Algorithms}, 20(3):545--580, 1996.

\bibitem{M04}
Bhubaneswar Mishra.
\newblock Computational real algebraic geometry.
\newblock In {\em Handbook of Discrete and Computational Geometry, 2nd Ed.},
  pages 743--764. Chapman and Hall/CRC, 2004.

\bibitem{MPVd16}
H.~{Nassajian Mojarrad}, T.~{Pham}, C.~{Valculescu}, and F.~{de Zeeuw}.
\newblock {Schwartz-Zippel} bounds for two-dimensional products.
\newblock {\em ArXiv e-prints}, 2016.
\newblock \href{https://arxiv.org/abs/1507.08181}{arXiv:1507.08181 [math.CO]}.

\bibitem{PS98}
J{\'{a}}nos Pach and Micha Sharir.
\newblock On the number of incidences between points and curves.
\newblock {\em Combinatorics, Probability {\&} Computing}, 7(1):121--127, 1998.

\bibitem{Alcala}
János Pach and Micha Sharir.
\newblock Combinatorial geometry with algorithmic applications – the
  {Alcal{\'{a}}} lectures.
\newblock {\em AMS, Providence}, 2009.

\bibitem{P10}
Mihai P{\u a}tra{\cb s}cu.
\newblock Towards polynomial lower bounds for dynamic problems.
\newblock In {\em {STOC}}, pages 603--610. {ACM}, 2010.

\bibitem{PS85}
Franco~P. Preparata and Michael~Ian Shamos.
\newblock {\em Computational Geometry - An Introduction}.
\newblock Texts and Monographs in Computer Science. Springer, 1985.

\bibitem{R72}
Michael~O. Rabin.
\newblock Proving simultaneous positivity of linear forms.
\newblock {\em J. Comput. Syst. Sci.}, 6(6):639--650, 1972.

\bibitem{RS15}
Orit~E. Raz and Micha Sharir.
\newblock The number of unit-area triangles in the plane: Theme and variations.
\newblock In {\em SoCG}, volume~34 of {\em LIPIcs}, pages 569--583, 2015.

\bibitem{RSZ15}
Orit~E. Raz, Micha Sharir, and Frank de~Zeeuw.
\newblock Polynomials vanishing on cartesian products: The
  {Elekes}-{Szab{\'{o}}} theorem revisited.
\newblock In {\em SoCG}, volume~34 of {\em LIPIcs}, pages 522--536, 2015.

\bibitem{RSZ16}
Orit~E. Raz, Micha Sharir, and Frank de~Zeeuw.
\newblock The elekes-szab{\'{o}} theorem in four dimensions.
\newblock {\em ArXiv e-prints}, 2016.
\newblock \href{https://arxiv.org/abs/1607.03600}{arXiv:1607.03600 [math.CO]}.

\bibitem{RSS15b}
Orit~E. Raz, Micha Sharir, and Ilya~D. Shkredov.
\newblock On the number of unit-area triangles spanned by convex grids in the
  plane.
\newblock {\em ArXiv e-prints}, 2015.
\newblock \href{https://arxiv.org/abs/1504.06989}{arXiv:1504.06989 [math.CO]}.

\bibitem{RSS14}
Orit~E. Raz, Micha Sharir, and J{\'{o}}zsef Solymosi.
\newblock Polynomials vanishing on grids: The {Elekes}-{R{\'{o}}nyai} problem
  revisited.
\newblock In {\em SoCG}, page 251. {ACM}, 2014.

\bibitem{RSS15}
Orit~E. Raz, Micha Sharir, and J{\'{o}}zsef Solymosi.
\newblock On triple intersections of three families of unit circles.
\newblock {\em Discrete {\&} Computational Geometry}, 54(4):930--953, 2015.

\bibitem{Sei74}
Abraham Seidenberg.
\newblock Constructions in algebra.
\newblock {\em Transactions of the AMS}, 197:273--313, 1974.

\bibitem{SY82}
J.~Michael Steele and Andrew Yao.
\newblock Lower bounds for algebraic decision trees.
\newblock {\em J. Algorithms}, 3(1):1--8, 1982.

\bibitem{T51}
Alfred Tarski.
\newblock A decision method for elementary algebra and geometry.
\newblock 1951.

\bibitem{Y81}
Andrew Yao.
\newblock A lower bound to finding convex hulls.
\newblock {\em J. {ACM}}, 28(4):780--787, 1981.

\bibitem{Y76}
David Yun.
\newblock On square-free decomposition algorithms.
\newblock In {\em {SYMSACC}}, pages 26--35, 1976.

\end{thebibliography}
